\documentclass[twocolumn,preprintnumbers,prd,nofootinbib,superscriptaddress]{revtex4-1}

\usepackage{amsmath}
\usepackage{amssymb}
\usepackage{amsthm}
\usepackage{graphicx}
\usepackage[label font={large}]{subfig}
\usepackage{changepage}
\usepackage{hyperref}
\usepackage{url}
\usepackage{breakurl}
\usepackage{mathtools,xparse}
\usepackage{verbatim}
\usepackage{dsfont}
\usepackage{xcolor}

\theoremstyle{definition}
\newtheorem{theorem}{Theorem}

\newtheorem{definition}{Definition}
\newtheorem*{definition*}{Definition}
\newtheorem{lemma}{Lemma}

\def\ket#1{| #1 \rangle}

\def\app#1#2{%
  \mathrel{%
    \setbox0=\hbox{$#1\sim$}%
    \setbox2=\hbox{%
      \rlap{\hbox{$#1\propto$}}%
      \lower1.1\ht0\box0%
    }%
    \raise0.25\ht2\box2%
  }%
}

\begin{document}

\title{Coarse-Graining Holographic States:\ A Semiclassical Flow in General Spacetimes}

\author{Chitraang Murdia}
\affiliation{Berkeley Center for Theoretical Physics, Department of Physics, University of California, Berkeley, CA 94720, USA}
\affiliation{Theoretical Physics Group, Lawrence Berkeley National Laboratory, Berkeley, CA 94720, USA}
\author{Yasunori Nomura}
\affiliation{Berkeley Center for Theoretical Physics, Department of Physics, University of California, Berkeley, CA 94720, USA}
\affiliation{Theoretical Physics Group, Lawrence Berkeley National Laboratory, Berkeley, CA 94720, USA}
\affiliation{Kavli Institute for the Physics and Mathematics of the Universe (WPI), The University of Tokyo Institutes for Advanced Study, Kashiwa 277-8583, Japan}
\author{Pratik Rath}
\affiliation{Berkeley Center for Theoretical Physics, Department of Physics, University of California, Berkeley, CA 94720, USA}
\affiliation{Theoretical Physics Group, Lawrence Berkeley National Laboratory, Berkeley, CA 94720, USA}


\begin{abstract}
Motivated by the understanding of holography as realized in tensor networks, we develop a bulk procedure that can be interpreted as generating a sequence of coarse-grained holographic states.
The coarse-graining procedure involves identifying degrees of freedom entangled at short distances and disentangling them.
This is manifested in the bulk by a flow equation that generates a codimension-1 object, which we refer to as the holographic slice.
We generalize the earlier classical construction to include bulk quantum corrections, which naturally involves the generalized entropy as a measure of the number of relevant boundary degrees of freedom.
The semiclassical coarse-graining results in a flow that approaches quantum extremal surfaces such as entanglement islands that have appeared in discussions of the black hole information paradox.
We also discuss the relation of the present picture to the view that the holographic dictionary works as quantum error correction.
\end{abstract}


\maketitle

\makeatletter
\def\l@subsection#1#2{}
\def\l@subsubsection#1#2{}
\makeatother
\tableofcontents


\section{Introduction}
\label{sec:intro}

The holographic principle, as embodied by the AdS/CFT correspondence, has led to a tremendous amount of progress in our understanding of quantum gravity.
In particular, the realization that entanglement plays a crucial role in generating bulk spacetime has put the holographic correspondence on much stronger footing~\cite{Ryu:2006bv,Ryu:2006ef,Hubeny:2007xt,VanRaamsdonk:2010pw,Swingle:2009bg}.
This has led to key insights about bulk reconstruction and subregion duality, culminating in entanglement wedge reconstruction~\cite{Czech:2012bh,Wall:2012uf,Headrick:2014cta,Jafferis:2015del,Dong:2016eik,Cotler:2017erl}.
Interestingly, several of these insights are quite general and do not seem to require an AdS setting in particular, and thus they can be used to understand features of holography in general spacetimes~\cite{Nomura:2016ikr,Nomura:2017npr,Nomura:2017fyh,Nomura:2018kji}.%
\footnote{An early work in this direction is the so-called surface/state correspondence~\cite{Miyaji:2015yva,Miyaji:2015fia}, of which the construction of Refs.~\cite{Nomura:2016ikr,Nomura:2017npr,Nomura:2017fyh,Nomura:2018kji} can be viewed as a covariant generalization. For other work on holography beyond AdS/CFT, see, e.g., Refs.~\cite{Alishahiha:2004md,Dong:2011uf,vanLeuven:2018pwv,Dong:2018cuv,Cooper:2018cmb,Gorbenko:2018oov,Geng:2020kxh}.}

A particular manifestation of the above ideas can be seen in tensor networks (TNs) that serve as useful toy models of holography~\cite{Swingle:2009bg,Pastawski:2015qua,Hayden:2016cfa,Donnelly:2016qqt,Qi:2017ohu}.
TNs prepare quantum states with a lot of structure and via the process of ``pushing'' the state generate a sequence of boundary states, each of which satisfies the Ryu-Takayanagi (RT) formula~\cite{Ryu:2006bv,Ryu:2006ef}.%
\footnote{We distinguish this from the Hubeny-Rangamani-Takayanagi (HRT) formula~\cite{Hubeny:2007xt} which applies in time-dependent spacetimes.}
This procedure involves disentangling certain short-distance degrees of freedom and coarse-grains the state by reducing it to one in a smaller effective Hilbert space.
Applying this procedure to a general smooth classical spacetime leads to a flow equation in the continuum limit as we shall review later \cite{Nomura:2018kji}.
The flow equation takes the form%
\footnote{The sign convention for the flow parameter $\lambda$ in this paper is opposite to that in Ref.~\cite{Nomura:2018kji}.}
\begin{equation}
    \frac{dx^\mu}{d\lambda} =\frac{1}{2}(\theta_k l^\mu+\theta_l k^\mu),
\nonumber
\end{equation}
where $x^\mu$ are the embedding coordinates of a codimension-2 surface $\sigma$ on which the holographic states are defined, and $\{ k^\mu, l^\mu \}$ are the future-directed null vectors orthogonal to $\sigma$, with $\theta_{k,l}$ being the classical expansions in the corresponding directions. This flow satisfies all the required properties for it to be interpreted as a disentangling procedure resulting in a sequence of coarse-grained states.

In this work, we go beyond the classical flow equation by including bulk quantum corrections.
In the TN picture, we include these effects by modifying the network such that it includes non-universal tensors/bonds as well as bonds connecting tensors nonlocally.
With this picture in mind, we develop a coarse-graining procedure analogous to the classical flow equation which pleasantly fits in with our understanding of holography.
In the continuum limit, the procedure leads to a flow equation similar to that in the classical case:
\begin{align}
    \frac{dx^\mu}{d\lambda} &=\frac{1}{2}(\Theta_k l^\mu+\Theta_l k^\mu),
\nonumber
\end{align}
where $\Theta_{k,l}$ now represent quantum expansions~\cite{Bousso:2015mna}.%
\footnote{We use a modified version of the quantum expansion which includes a bulk entropy contribution from an exterior region as described in Sec.~\ref{sec:main}.}
This is our primary result.
It is in line with many results in which quantum corrections are included by replacing the area $\mathcal{A}/4G_N$ with the generalized entropy $S_{\text{gen}}$~\cite{Bousso:2015mna,Bousso:2015eda,Wall:2018ydq,Bousso:2019dxk,Bousso:2019var}.
Though motivated by TNs, which often face issues in describing time-dependent situations, our procedure can be applied quite generally.
In fact, we obtain consistent descriptions in general time-dependent spacetimes.

Another important progress in understanding holography is the view that the holographic dictionary works as quantum error correction~\cite{Almheiri:2014lwa,Harlow:2016vwg,Akers:2018fow,Dong:2018seb}, where a small Hilbert space of semiclassical bulk states is mapped isometrically into a larger boundary Hilbert space.
In our framework, this picture arises after considering a collection of states over which we want to build a low energy bulk description.
Choosing such a collection is equivalent to erecting a code subspace.
We argue that while there is no invariant choice of code subspace in a general time-dependent spacetime, our framework gives a natural choice(s) determined by the coarse-graining procedure.
This procedure leads to a one-parameter family of ``dualities'' depending on the amount of coarse-graining performed, providing an improved understanding of the holographic dictionary in general spacetimes.


\subsection*{Overview}

In Sec.~\ref{sec:framework}, we first establish the framework in which we are working.
We explain how quantum corrections affect the description of holography in general spacetimes and the associated HRT formula.
In Sec.~\ref{sec:review}, we review the classical flow equation.
In Sec.~\ref{sec:motive}, we motivate our coarse-graining procedure with a toy model of TNs, elucidating how features of a state relevant for the quantum-level consideration are represented there.

In Sec.~\ref{sec:main}, we present our main result, i.e.\ the procedure of performing the flow in the bulk at the quantum level, which corresponds to moving the holographic boundary.
We also discuss properties of this flow indicating that it corresponds to a coarse-graining of holographic states.
We elucidate that the way the flow ends can be used as an indicator of qualitative features of the boundary state describing a given spacetime, using the example of a collapse-formed evaporating black hole.
In Sec.~\ref{sec:QEC}, we discuss how the picture of quantum error correction may be implemented in our framework.
Finally, conclusions are given in Sec.~\ref{sec:concl}.


\section{Framework}
\label{sec:framework}

In this work, we follow and further develop the framework of holography for general spacetimes proposed in Ref.~\cite{Nomura:2016ikr}.
In this framework, we consider an arbitrary spacetime $\mathcal{M}$ and posit the existence of a dual ``boundary'' theory that lives on a holographic screen~\cite{Bousso:2002ju}, which is a codimension-1 hypersurface $H$ embedded in $\mathcal{M}$. This hypersurface is foliated by marginally trapped/anti-trapped codimension-2 surfaces called leaves, which we denote by $\sigma$.
A marginally trapped/anti-trapped surface $\sigma$ is defined by the property that $\sigma$ has classical expansion $\theta=0$ in one of the orthogonal null directions.
The proposal is that the boundary theory describes everything in the ``interior'' of $H$, and states of the theory are naturally defined on the leaves $\sigma$, which provide a preferred foliation of $H$ into constant time surfaces.
Based on the covariant entropy bound, it is expected that the boundary theory effectively possesses $\mathcal{A}(\sigma)/4G_N$ degrees of freedom, where $\mathcal{A}(\sigma)$ is the area of a leaf.
The AdS/CFT correspondence can be viewed as a special case of this duality, where the holographic screen is sent to the conformal boundary of AdS.

Given this setup, it was shown in Ref.~\cite{Sanches:2016sxy} that the HRT formula for computing entanglement entropy can be applied consistently using a maximin procedure~\cite{Wall:2012uf}; i.e., for any subregion $A$ of a leaf
\begin{align}\label{eq:HRT}
    S(A)&= \frac{\mathcal{A}(\gamma_A)}{4G_N},
\end{align}
where $S(A)$ is the von~Neumann entropy of the reduced density matrix on subregion $A$, and $\mathcal{A}(\gamma_A)$ is the area of the HRT surface $\gamma_A$ of $A$.
The entanglement wedge, denoted by $\text{EW}(A)$, is defined as the bulk domain of dependence of any bulk partial Cauchy slice $\Sigma_A$ with $\partial\Sigma_A = A \cup \gamma_A$, which is often called the homology surface.
The entropies obtained by the above procedure can be shown to satisfy all the basic properties of von~Neumann entropy and are consistent with more constraining inequalities satisfied by holographic states in AdS/CFT~\cite{Headrick:2007km,Hayden:2011ag,Bao:2015bfa}.

Now, in order to generalize this framework to the quantum level, we can follow the simple guiding principle of replacing $\mathcal{A}/4G_N$ with the generalized entropy $S_{\text{gen}}$ to include quantum corrections in the bulk~\cite{Faulkner:2013ana,Engelhardt:2014gca}
\begin{align}\label{eq:Sgen}
    \frac{\mathcal{A}}{4G_N} \,\rightarrow\, S_{\text{gen}} = \frac{\mathcal{A}}{4G_N} + S_{\text{bulk}},
\end{align}
where $S_{\text{bulk}}$ is the von~Neumann entropy of the bulk reduced density matrix on the homology surface which is appropriately modified at the quantum level.
This is motivated by various examples in which this naturally works~\cite{Bousso:2015mna,Bousso:2015eda,Wall:2018ydq,Bousso:2019dxk,Bousso:2019var}.
Furthermore, $S_{\text{gen}}$ is a natural quantity because it is a quantity that is renormalization scheme independent, and hence is expected to be associated with fundamental degrees of freedom in the UV theory~\cite{Susskind:1994sm,Solodukhin:2011gn,Wall:2018ydq}.

The generalization to include bulk quantum corrections requires a refined understanding of the holographic duality which we now turn to.
First, we note that the global description of a state involves both the interior and exterior portions of the holographic screen~\cite{Bousso:2002ju}.
Although the generalization of the HRT formula we describe applies to the interior region, it will be important to keep track of the exterior as well.

\begin{figure}[t]
\centering
  \includegraphics[clip,width=0.55\columnwidth]{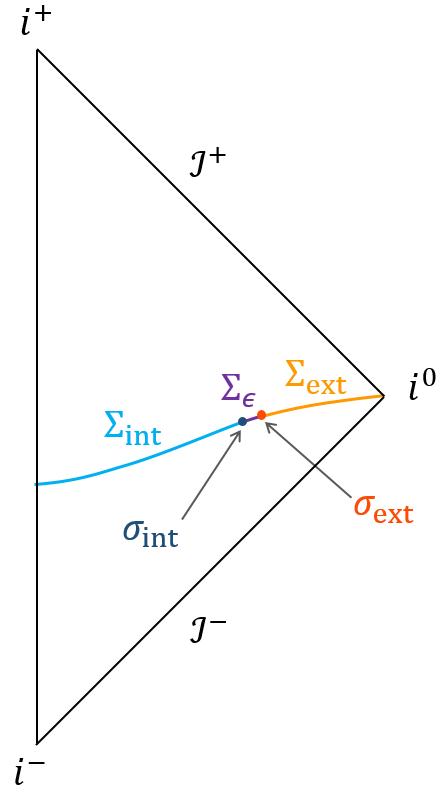}
\caption{The leaf $\sigma$ is split into $\sigma_{\text{int}} \cup \sigma_{\text{ext}}$ such that $\sigma_{\text{int}}$ and $\sigma_{\text{ext}}$ are separated by a small regulating region $\Sigma_{\epsilon}$ on a Cauchy slice $\Sigma$.
This induces a division of the Cauchy slice as $\Sigma = \Sigma_{\text{int}} \cup \Sigma_{\epsilon} \cup \Sigma_{\text{ext}}$.
We define the location of the holographic screen by requiring that it is marginally quantum trapped/anti-trapped under variations of $\sigma_{\text{int}}$.}
\label{fig:regulate}
\end{figure}

Next, the quantum extension of Eq.~(\ref{eq:Sgen}) implies that the location of the screen on which the holographic theory is defined needs to be shifted accordingly.
Let us consider a specific global bulk state.
We propose that the boundary theory describing the dynamics of this state lives on a modified version of a Q-screen $H'$~\cite{Bousso:2015eda}, rather than a classical holographic screen $H$.
A Q-screen is defined as a codimension-1 hypersurface foliated by quantum marginally trapped/antitrapped surfaces, i.e.\ surfaces that have the quantum expansion $\Theta=0$ in one of the orthogonal null directions.
Usually the quantum expansion is defined by including a contribution from the von Neumann entropy of the interior or exterior region of a leaf which is a simple codimension-2 surface.
In this paper, however, we consider that leaf $\sigma$ is given by $\sigma_{\text{int}} \cup \sigma_{\text{ext}}$ such that $\sigma_{\text{int}}$ and $\sigma_{\text{ext}}$ are split by a small regulating region $\Sigma_{\epsilon}$ on a Cauchy slice $\Sigma$ as seen in Fig.~\ref{fig:regulate}.
This induces a division of the Cauchy slice as $\Sigma = \Sigma_{\text{int}} \cup \Sigma_{\epsilon} \cup \Sigma_{\text{ext}}$.
Now, we define the generalized entropy of $\sigma$ to be 
\begin{align}
    S_{\text{gen}}(\sigma)&=\frac{\mathcal{A(\sigma)}}{4G_N}+ S_{\text{bulk}}(\Sigma_{\text{int}} \cup \Sigma_{\text{ext}}).
\label{eq:def-S_gen}
\end{align}
With this definition of generalized entropy, one can define a quantum expansion $\Theta$ by the variation of $S_{\text{gen}}(\sigma)$ under deformations of $\sigma_{\text{int}}$ while holding $\sigma_{\text{ext}}$ fixed.
Using this definition of $\Theta$, we can locate marginally trapped/anti-trapped surfaces self-consistently for any given $\epsilon$.
The location $H'$ of the holographic screen is a Q-screen defined using this definition of $\Theta$ in the limit $\epsilon \to 0$.

Generalizing the HRT formula of Eq.~(\ref{eq:HRT}), we postulate that the von~Neumann entropy of a subregion $A$ on the leaf  $\sigma$ of a Q-screen can be computed as
\begin{align}\label{eq:QES}
    S(A)&= \frac{\mathcal{A}(A \cup \Gamma_A)}{4G_N} + S_{\text{bulk}}(\Sigma_A),
\end{align}
where $\Gamma_A$ is the minimal quantum extremal surface (QES)~\cite{Engelhardt:2014gca}, and $\Sigma_A$ is the homology surface with $\partial\Sigma_A = A \cup \Gamma_A$.
To find $\Gamma_A$, we can use a maximin procedure at the quantum level~\cite{Akers:2019lzs} applied to general spacelike surfaces containing $\sigma$.
We treat $\sigma_{\text{ext}}$ as a single unit that cannot be further divided into subregions, which must be either included in or excluded from $A$.
We assume that the leaf $\sigma$ is convex where $\sigma_{\text{ext}}$ is treated as an indivisible unit.
Thus, $S(A)$ obtained by Eq.~(\ref{eq:QES}) satisfies properties required for it to be interpreted as the von~Neumann entropy of the density matrix of subregion $A$.
With this assumption, we will show that the same applies to any renormalized leaf $\sigma(\lambda)$ obtained from $\sigma$ by our coarse-graining procedure.

In AdS/CFT, the regime of validity of the quantum extremal surface formula, Eq.~\eqref{eq:QES}, has been suggested to be all orders in $G_N$~\cite{Engelhardt:2014gca}.
However, there are subtleties with the definition of entanglement entropy for gravitons which have not been completely resolved (see e.g.\ Ref.~\cite{Bousso:2015mna}).
At the least, we expect the formula to hold at $O(1)$, where it can already lead to a surface different from that obtained by using Eq.~\eqref{eq:HRT}~\cite{Penington:2019npb,Almheiri:2019psf,Almheiri:2019hni,Penington:2019kki,Almheiri:2019qdq}.
In order to avoid subtleties with gravitons, one could consider a setup with bulk matter having a large central charge $c$ so that the graviton contribution is subleading in $1/c$ expansion.
We expect a similar regime of validity for Eq.~\eqref{eq:QES} in general spacetimes.

Note that in Eq.~(\ref{eq:QES}) we have included the area contribution from $A$ in addition to that from $\Gamma_A$.
This is required if there is a spacetime region outside the leaf, as is the case in generic spacetimes.
In this case, $S_{\text{bulk}}(\Sigma_A)$ receives a contribution from entanglement of bulk fields across $A$, which is divergent.
This divergence is then canceled with that in the area contribution from $A$ in the first term, making $S(A)$ well defined.
The same applies to the AdS/CFT case if we impose transparent boundary conditions near the boundary which lead to kinetic terms coupling the interior and exterior of AdS space~\cite{Penington:2019npb,Almheiri:2019psf}.
In fact, the classical formula in Eq.~(\ref{eq:HRT}) must also have the contribution from the boundary subregion area, $\mathcal{A}(\gamma_A) \rightarrow \mathcal{A}(A \cup \gamma_A)$, in these cases, although this does not affect the minimization leading to $\gamma_A$ and hence the result of Ref.~\cite{Nomura:2018kji}.

There are special cases in which the area contribution from $A$---as well as the corresponding contribution from $S_{\text{bulk}}(\Sigma_A)$---is absent.
This occurs when the spacetime outside the leaf is ``absent,'' as is the case if Dirichlet boundary conditions are imposed on the Q-screen, or if reflective boundary conditions are imposed in AdS/CFT.
Even in this case, however, our coarse-graining procedure---which corresponds to moving the leaf portion $\sigma_{\text{int}}$ inward---induces the area contribution from $A$ on a moved (i.e.\ renormalized) leaf $\sigma_{\text{int}}(\lambda)$, reflecting the fact that the spacetime continues across $\sigma_{\text{int}}(\lambda)$.%
\footnote{This is different from what has been done in the AdS/CFT literature in the context of $TT$ deformations~\cite{McGough:2016lol,Donnelly:2018bef,Caputa:2019pam,Banerjee:2019ewu,Murdia:2019fax,Lewkowycz:2019xse},
which corresponds to (re)imposing Dirichlet boundary conditions at each step in the coarse-graining, i.e.\ at $\sigma_{\text{int}}(\lambda)$ for all $\lambda$.}


\section{Review of Classical Flow}
\label{sec:review}

In previous work~\cite{Nomura:2018kji}, it was shown that a coarse-graining procedure motivated by TNs can be defined in the bulk at the level of classical geometry.
Here we review this construction, which allows us to elucidate a generalization to include bulk quantum corrections.

\begin{figure}[t]
  \includegraphics[clip,width=0.8\columnwidth]{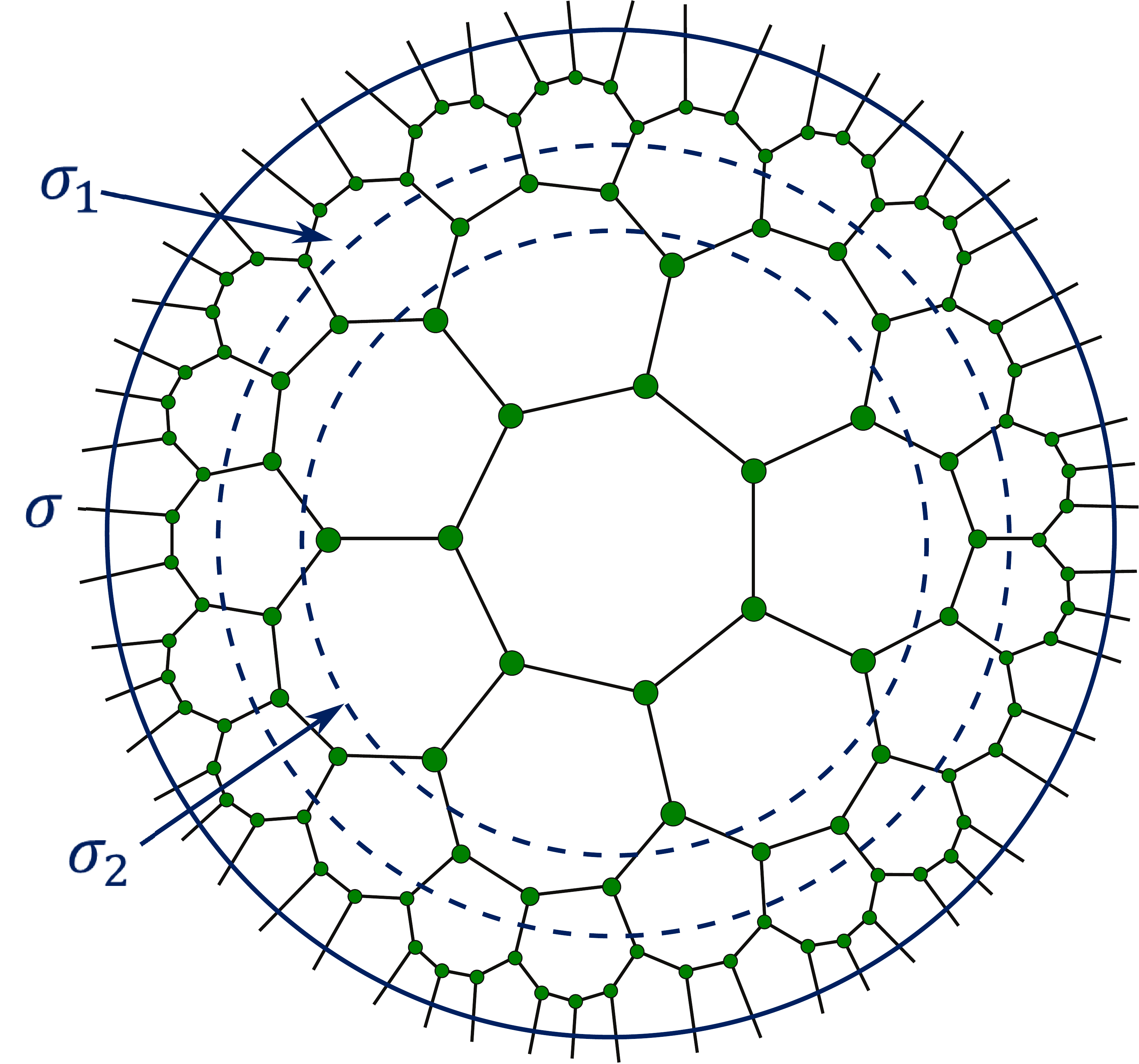}
\caption{A TN defines a boundary state in the Hilbert space $\mathcal{H}_{\sigma}$ at the outer legs. One can, however, also consider ``coarse-grained'' states defined at inner layers, e.g.\ states defined in Hilbert spaces $\mathcal{H}_{\sigma_1}$ and $\mathcal{H}_{\sigma_2}$.}
\label{fig:tennet}
\end{figure}

A key idea is to realize that a TN defines a sequence of states that can be generated by including fewer tensors, layer by layer, as shown in Fig.~\ref{fig:tennet}.
For example, one can consider a state defined on the outermost legs which lives in Hilbert space $\mathcal{H}_{\sigma}$.
A coarse-grained version of this state can then be given by a smaller TN that is obtained by peeling off the outermost layer. This state lives in a smaller Hilbert space $\mathcal{H}_{\sigma_1}$, and the TN provides an isometric map from $\mathcal{H}_{\sigma_1}$ to $\mathcal{H}_{\sigma}$.
The sequence can then continue, giving a series of Hilbert spaces $\mathcal{H}_{\sigma_2}$, $\mathcal{H}_{\sigma_3}$, and so on.

\begin{figure}[t]
  \subfloat[]{%
  \includegraphics[clip,width=0.8\columnwidth]{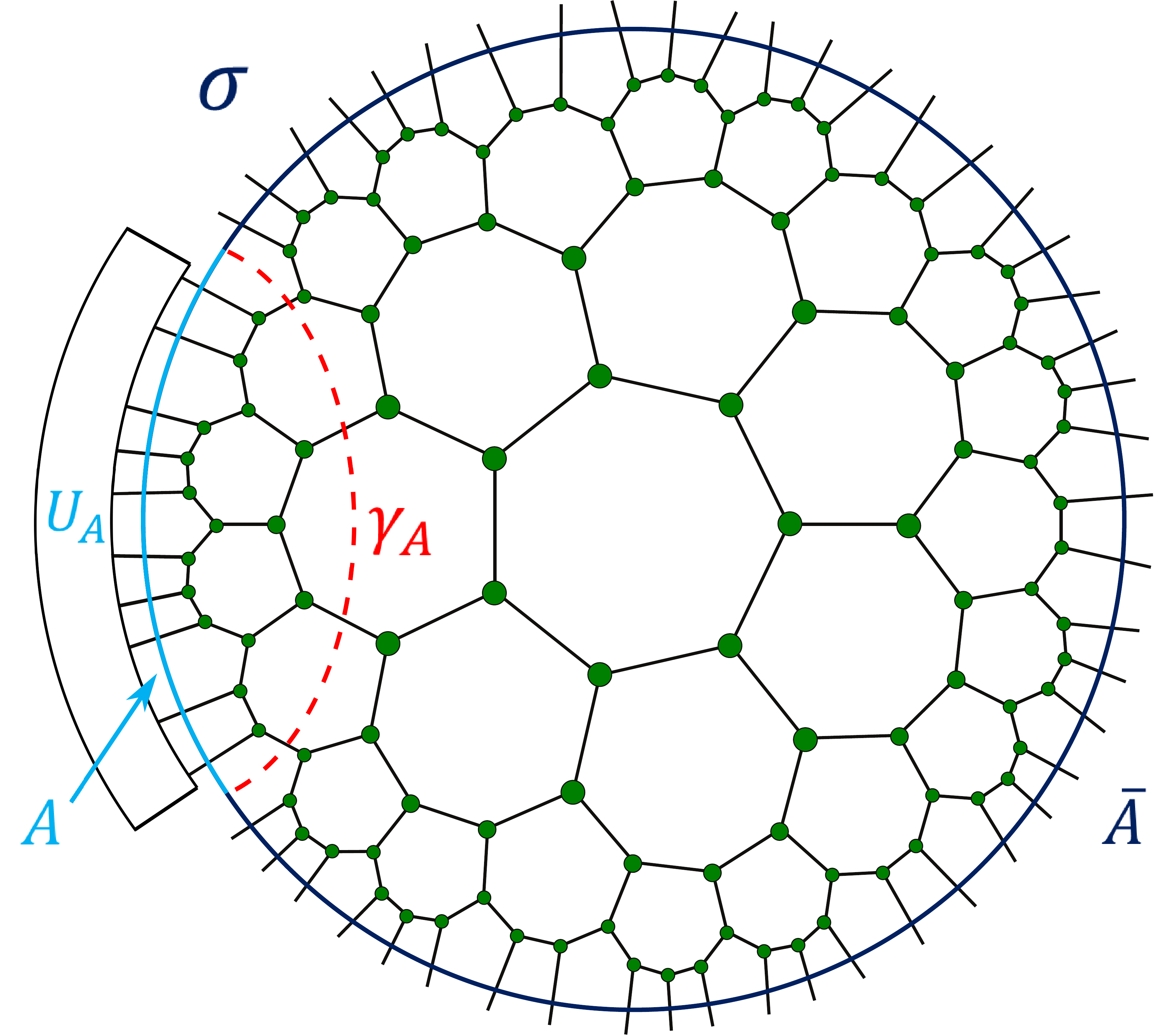}
  \label{fig:tencut1}}

  \subfloat[]{%
  \includegraphics[clip,width=0.8\columnwidth]{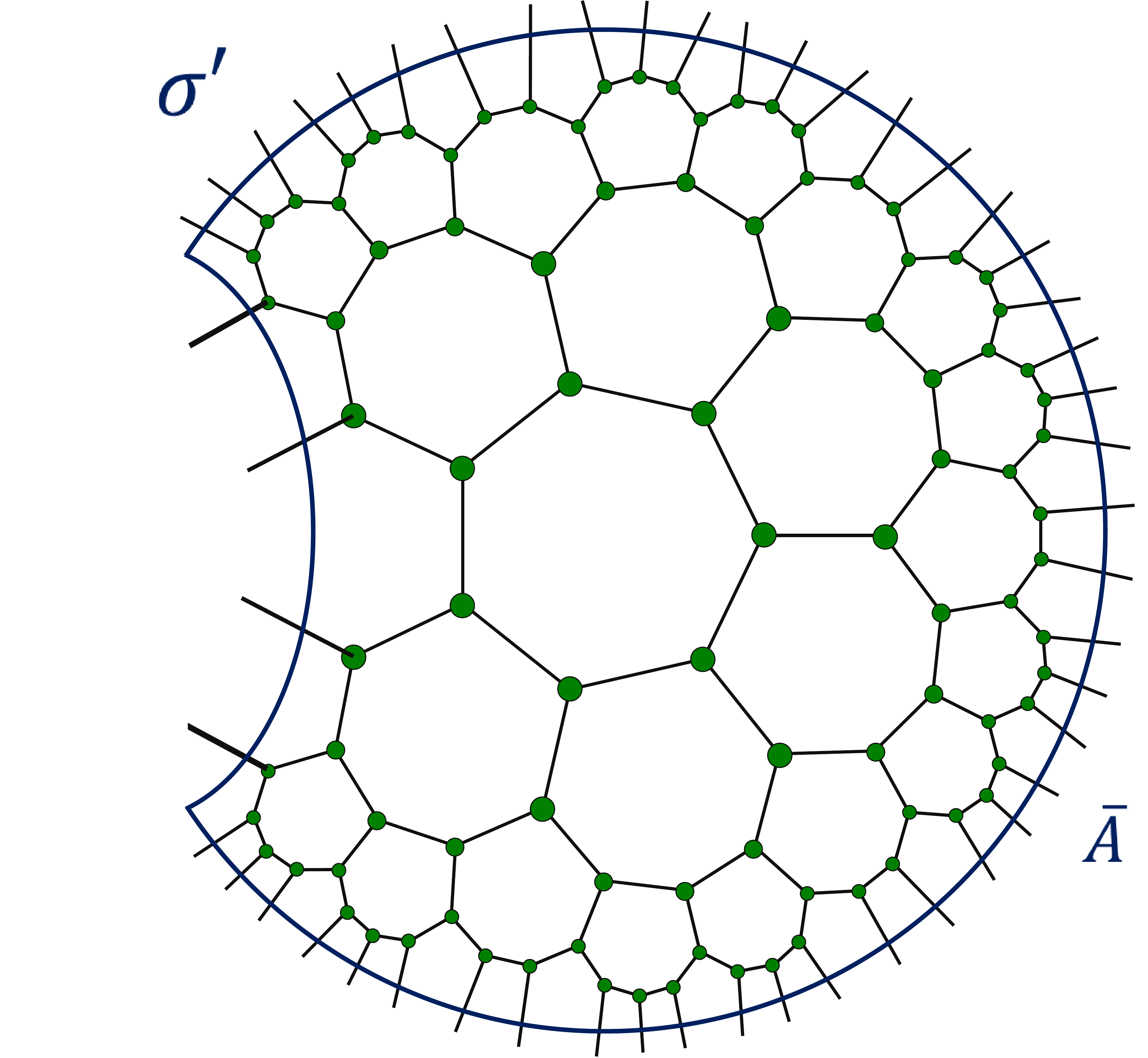}%
  \label{fig:tencut2}}
\caption{a) The von~Neumann entropy of subregion $A$ is computed by the minimal cut $\gamma_A$ that splits the TN into two parts containing $A$ and $\overline{A}$ respectively.
b) By applying a local unitary on $A$, we can find maximally entangled legs across $\gamma_A$, which serve as a bottleneck for the entanglement between $A$ and $\overline{A}$. }
\label{fig:tencut}
\end{figure} 

In fact, this peeling-off procedure can be decomposed further into smaller steps.
For any subregion $A$ of a given boundary, there is an isometric map from the in-plane legs at the RT surface $\gamma_A$ to the boundary legs in $A$~\cite{Pastawski:2015qua,Hayden:2016cfa}.
This implies that a particular subspace of the boundary subregion legs, corresponding to the in-plane legs at $\gamma_A$, is maximally entangled with the complementary subregion $\overline{A}$ via $\gamma_A$, which acts as an entanglement ``bottleneck''; see Fig.~\ref{fig:tencut}.
All the other subregion legs can be disentangled by applying unitary $U_A$ that acts locally within $A$.

Therefore, if one is to preserve only long range entanglement while getting rid of short range entanglement, one could compress the state down to that defined at the legs of the surface $\sigma' = \gamma_A \cup \overline{A}$.
This reduces the size of the effective Hilbert space, mapping a pure state in the larger boundary Hilbert space $\mathcal{H}_{\sigma}$ to a pure state in a smaller boundary Hilbert space $\mathcal{H}_{\sigma'}$.
This can be done by considering small subregions of $\sigma$ and truncating the TN to end at $\sigma'$.
One useful way to interpret this step is that we are retaining the complementary entanglement wedge $\text{EW}(\overline{A})$.
This step can then be repeated multiple times to generate a sequence of states, all of which are increasingly coarse-grained.%
\footnote{We note that this is similar to the construction suggested in Refs.~\cite{Bao:2018pvs,Bao:2019fpq}, although here we directly use the TN description and its fine structure, as opposed to constructing the TN using information about the boundary state such as entanglement of purification~\cite{Takayanagi:2017knl,Nguyen:2017yqw} or reflected entropy~\cite{Dutta:2019gen}.}

\begin{figure}[t]
  \includegraphics[clip,width=0.9\columnwidth]{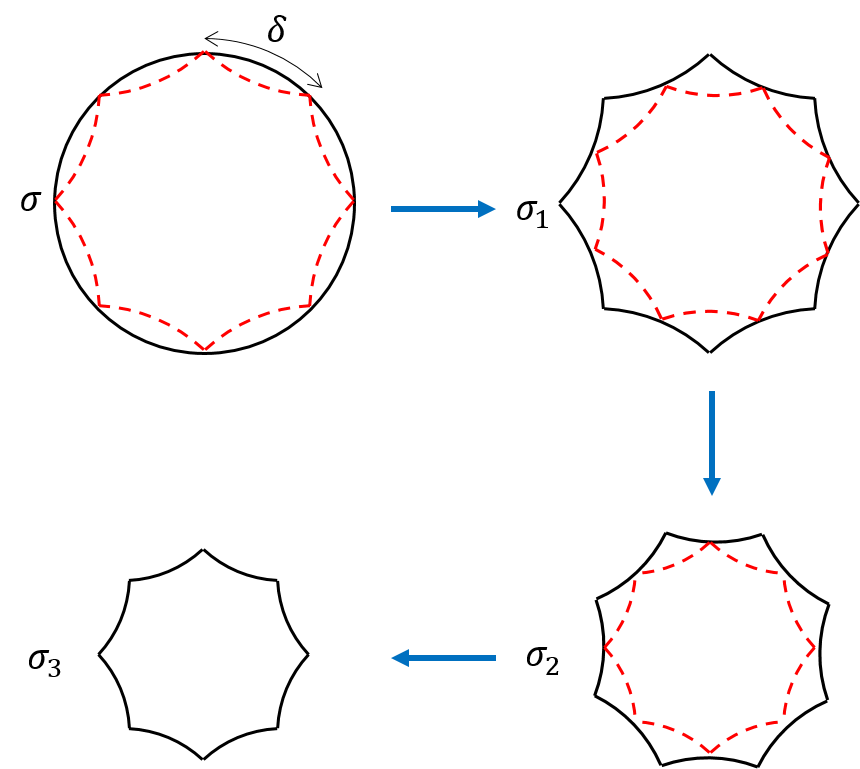}
\caption{A sequence of coarse-graining steps. At each step, we consider infinitesimal subregions of size $\delta$ ($\rightarrow 0$) and reduce the spacetime region to their respective complementary entanglement wedges.}
\label{fig:leafren}
\end{figure}

Now, we apply this idea to a general spacetime $\mathcal{M}$ by considering infinitesimal subregions $A$ of size $\delta$ ($\rightarrow 0$) on the boundary leaf $\sigma$.
In order to coarse-grain, we find the HRT surface $\gamma_A$ and reduce the accessible spacetime region to the complementary entanglement wedge, $\text{EW}(\overline{A})$.
Repeating this multiple times involves shrinking the domain of dependence at each step by finding new HRT surfaces anchored to infinitesimal subregions as seen in Fig.~\ref{fig:leafren}.

In the continuum limit, this reduces to the original construction of Ref.~\cite{Nomura:2018kji}; see Fig.~\ref{fig:continuum}.
Here, we consider the intersection of the complementary entanglement wedges $\text{EW}(\overline{A_p})$ for infinitesimal subregions $A_p$, centered around arbitrary points on the leaf, denoted by $p$.
This leads to a new domain of dependence $R(\sigma)$,
\begin{align}\label{eq:intersect}
    R(\sigma) = \cap_{p}\, \text{EW}(\overline{A_p}),
\end{align}
which can be interpreted as defining the state on a new ``renormalized'' leaf $\sigma_1$ such that the domain of dependence of $\sigma_1$ is $R(\sigma)$, i.e.\ $D(\sigma_1) = R(\sigma)$.%
\footnote{The domain of dependence of a closed codimension-2 surface is defined as the domain of dependence of a spacelike hypersurface enclosed by the surface.}
The HRT prescription can be shown to consistently apply for subregions on this renormalized leaf as well, owing to the fact that it is still a convex surface~\cite{Nomura:2018kji}.
Thus, we may interpret this as the spacetime continuum version of the procedure yielding the sequence of states described above using TNs.

\begin{figure}[t]
  \includegraphics[clip,width=0.8\columnwidth]{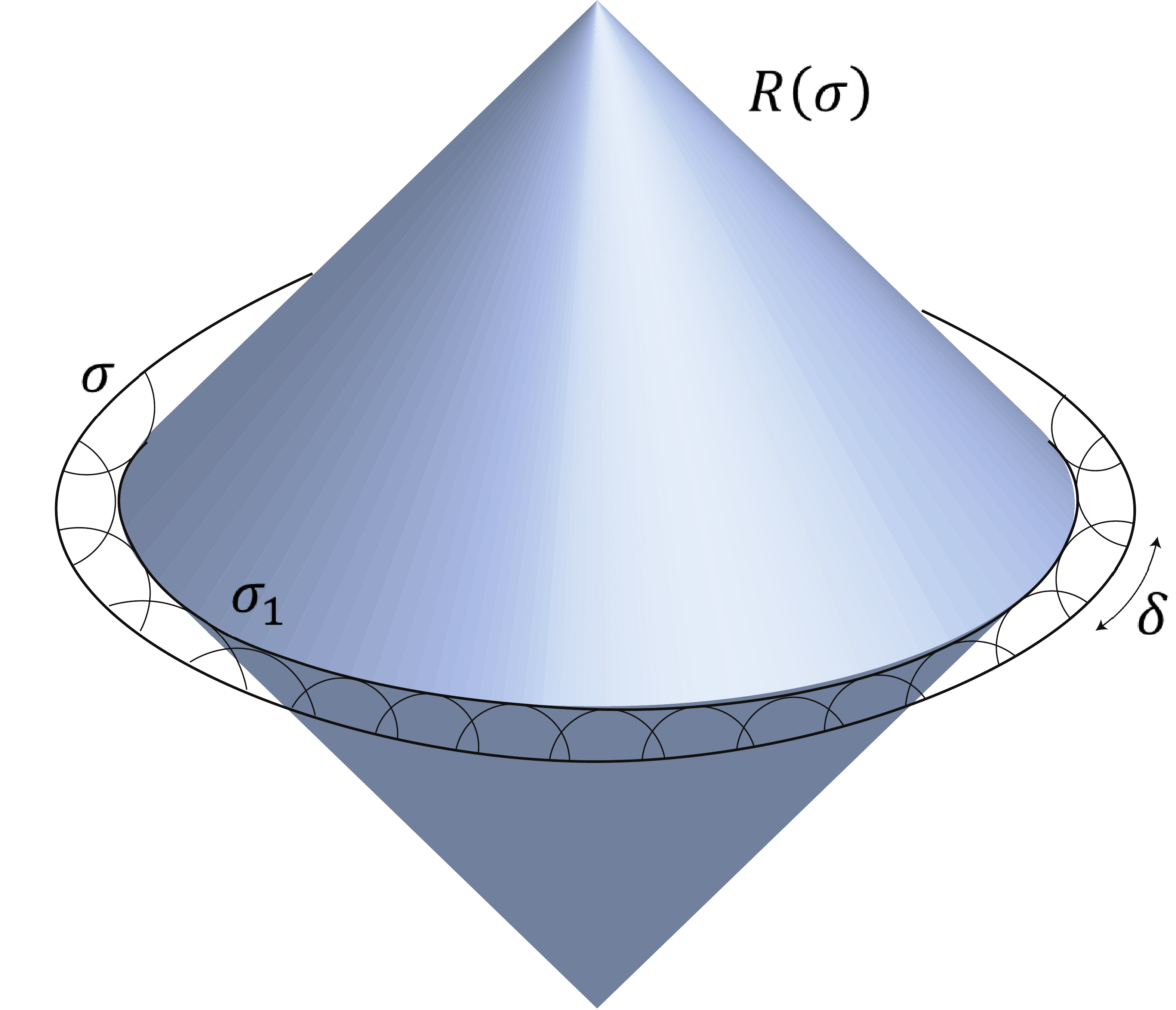}
\caption{Coarse-graining over infinitesimal subregions on $\sigma$ can be performed by considering the intersection of complementary entanglement wedges. This leads to a domain of dependence $R(\sigma)$ which corresponds to a new renormalized leaf $\sigma_1$.}
\label{fig:continuum}
\end{figure}

This continuum procedure can be written in terms of a flow equation for the leaf $\sigma(\lambda)$, which is interpreted as a Lorentzian mean curvature flow:
\begin{align}\label{eq:classicalflow}
    \frac{dx^\mu}{d\lambda}&= \frac{1}{2}(\theta_k l^\mu + \theta_l k^\mu),
\end{align}
where $x^\mu$ are the embedding coordinates of $\sigma(\lambda)$, $\{ k^\mu, l^\mu \}$ are the future-directed null vectors orthogonal to $\sigma(\lambda)$, normalized such that $k \cdot l = -2$, and $\theta_{k,l}$ are the classical expansions in the $k, l$ directions, respectively.
The sequence of renormalized leaves spans a codimension-1 hypersurface, which was termed the holographic slice.
In particular, it is a partial Cauchy slice of the bulk domain of dependence, $D(\sigma)$, of the original leaf $\sigma$.

It was shown that the flow described above satisfies various interesting properties that are consistent with the coarse-graining interpretation.
These include the fact that the area of the leaf $\sigma(\lambda)$ decreases monotonically with $\lambda$, implying that the number of degrees of freedom in the effective Hilbert space $\mathcal{H}_{\rm eff}(\sigma(\lambda))$ decreases as we flow.
By choosing statistically isotropic subregions with random shapes, one can obtain a preferred holographic slice that preserves the symmetries of the system.
Alternatively, by varying flow rates along the transverse directions, one could get a range of different, but gauge equivalent, slices of $D(\sigma)$.


\section{Motivation from Tensor Networks}
\label{sec:motive}

Having the classical construction in hand, we now describe how to generalize it to include bulk quantum corrections.
Let us take a specific state defined on a leaf of a Q-screen.
We want to understand how coarse-graining of this state works using the TN picture.

We expect that the state is still modeled by a TN at the quantum level.
In order to represent the effect of bulk quantum corrections appropriately, this TN must include two additional features compared with the classical case.
First, tensors and bonds used in the network should in general not be all ``featureless,'' i.e.\ all the tensors being the same and connected by maximally entangling bonds, as was the case in simple perfect tensor~\cite{Pastawski:2015qua} or random tensor~\cite{Hayden:2016cfa} networks.
Reflecting the existence of excitations of bulk low-energy fields, tensors and/or bonds must have a non-universal structure representing such excitations.
This generally makes the network not fully isometric.
Second, since bulk low-energy quantum fields can have long-range entanglement, corresponding to $S_\mathrm{bulk}$ in $S_\mathrm{gen}$, there should be longer bonds connecting non-nearest-neighbor tensors, although the number of such non-local bonds (or more precisely, the total dimension associated with them) is suppressed generally as the bonds become longer.
A typical TN of this sort is depicted in Fig.~\ref{fig:q-TN}.

\begin{figure}[t]
  \includegraphics[clip,width=0.8\columnwidth]{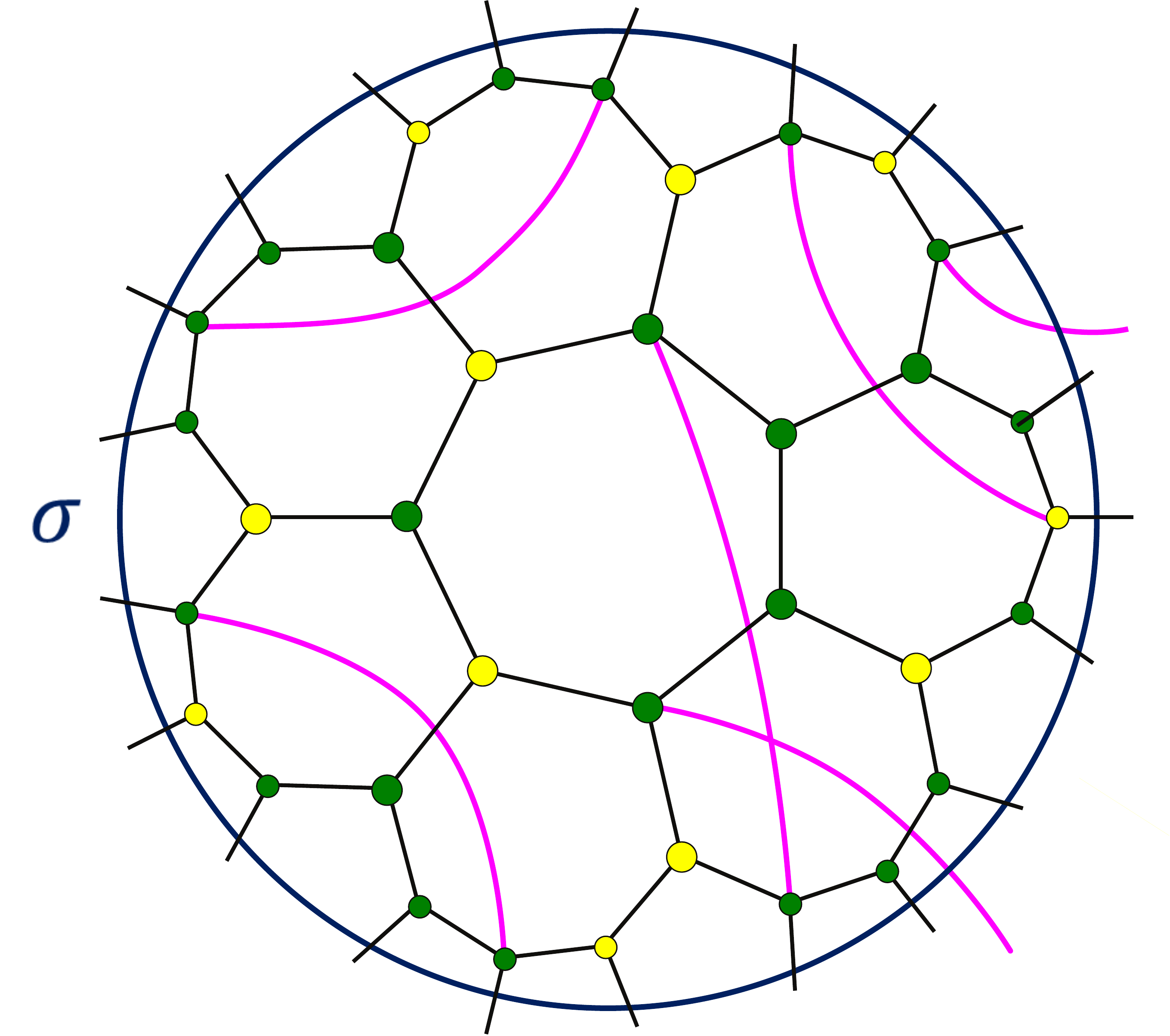}
\caption{A TN representing a state at the quantum level has tensors that are not universal (yellow) and bonds that connect tensors nonlocally (pink).}
\label{fig:q-TN}
\end{figure}

Once we have a TN representation of the state, the scenario in Sec.~\ref{sec:review} can be generalized in a relatively simple manner.
To do so, we must first establish how to compute the entropy of a boundary subregion, following the formula in Eq.~(\ref{eq:QES}).
In general, the boundary legs consist of both the shortest, local bonds and longer, nonlocal bonds cut by the boundary.
When we compute the entropy of subregion $A$ of $\sigma_{\text{int}}$, i.e.\ a subset of these legs, we must minimize
\begin{equation}
    \tilde{S}_\mathrm{gen}(A, X_A) = \frac{\mathcal{A}(A \cup X_A)}{4G_N} + S_\mathrm{bulk}(\Xi_A)
\label{eq:S_gen-2}
\end{equation}
over all surfaces $X_A$ anchored to the boundary of $A$, where $\Xi_A$ is the homology surface with $\partial\Xi_A = A \cup X_A$.
In this expression, the area term represents the contribution from the shortest bonds, while $S_\mathrm{bulk}(\Xi_A)$ from longer bonds, cut by $\partial\Xi_A$.
In short, $\tilde{S}_\mathrm{gen}(A, X_A)$ is given by the entropy of all the bonds connecting tensors inside and outside $\partial\Xi_A$, regardless of their lengths; see Fig.~\ref{fig:q-wedg}.
(This reflects the fact that the precise way to separate the contributions from local and nonlocal bonds is arbitrary and does not have an invariant meaning.)

\begin{figure}[t]
\centering
\includegraphics[clip,width=0.8\columnwidth]{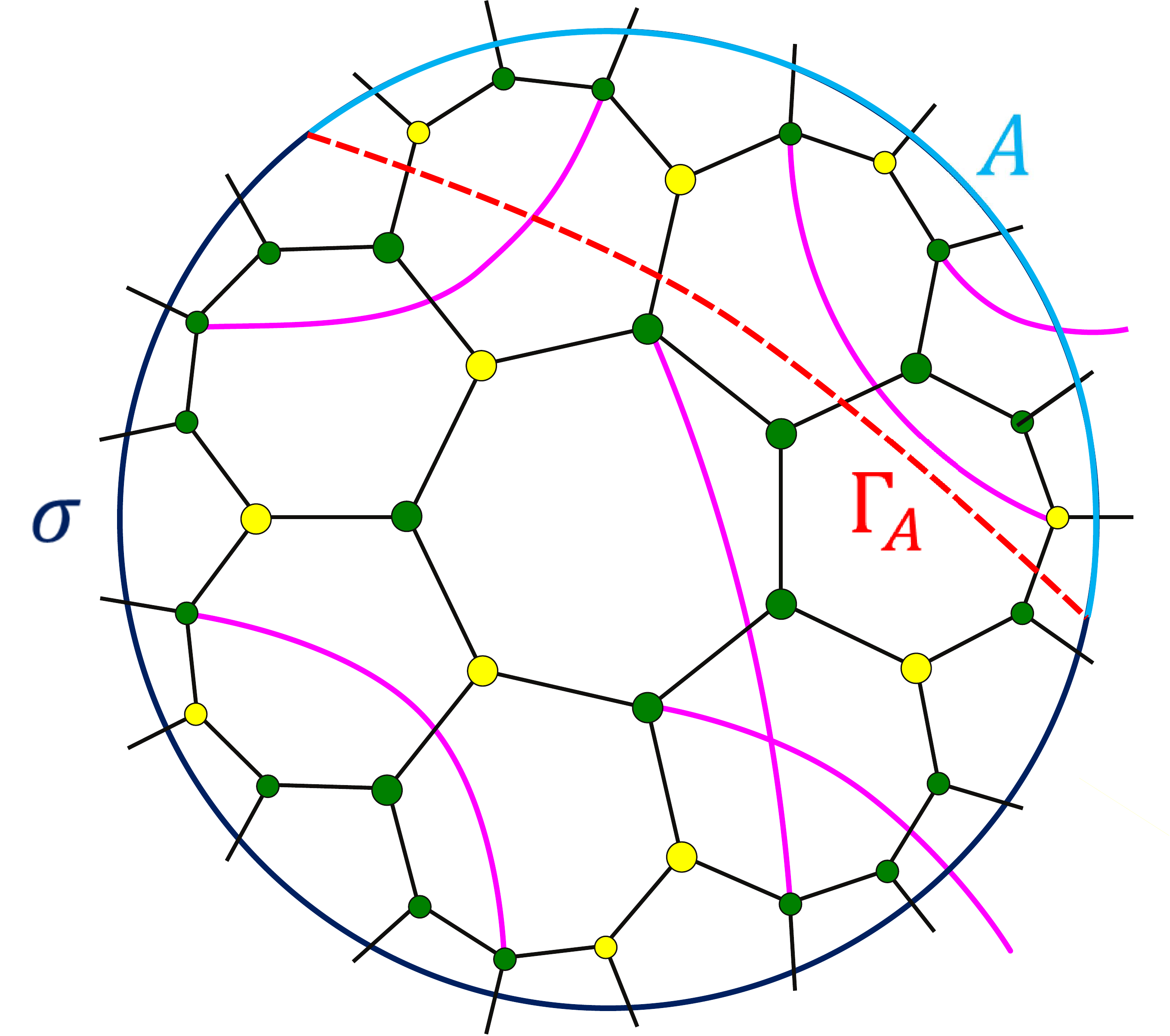}
\caption{The minimal QES $\Gamma_A$ of subregion $A$ is determined by minimizing the entropy of all bonds connecting tensors inside and outside $A \cup \Gamma_A$ (consisting of 12 black and 2 pink bonds in the figure).}
\label{fig:q-wedg}
\end{figure}

With this prescription, we can follow the analysis in Sec.~\ref{sec:review} and coarse-grain the region $A$ of the boundary state by removing the bulk regions that are entangled with $A$, i.e.\ by reducing the TN to a smaller one giving a state on
\begin{align}\label{eq:W}
    \sigma' = \overline{A} \cup \Gamma_{\overline{A}}.
\end{align}
Note that here the complement $\overline{A}$ of $A$ is defined as that on the entire leaf $\sigma = \sigma_{\text{int}} \cup \sigma_{\text{ext}}$, not on $\sigma_{\text{int}}$.
The QES $\Gamma_{\overline{A}}$ is given by the surface $X_{\overline{A}}$ minimizing
\begin{equation}
\label{eq:Sgen}
  \tilde{S}_\mathrm{gen}(\overline{A}, X_{\overline{A}}) = \frac{\mathcal{A}(\overline{A} \cup X_{\overline{A}})}{4G_N} + S_\mathrm{bulk}(\Xi_{\overline{A}}),
\end{equation}
where
\begin{equation}
    \partial X_{\overline{A}} = \partial \overline{A} = \partial \overline{A}_{\text{int}}
\end{equation}
with $\overline{A}_{\text{int}}$ the complement of $A$ on $\sigma_{\text{int}}$ (i.e.\ $\overline{A} = \overline{A}_{\text{int}} \cup \sigma_{\text{ext}}$), and
$\partial \Xi_{\overline{A}} = \overline{A} \cup X_{\overline{A}}$, meaning that
\begin{equation}
    \Xi_{\overline{A}} = \Xi_{\overline{A},\text{int}} \cup \Sigma_{\text{ext}}.
\end{equation}
Here, $\partial \Xi_{\overline{A},\text{int}} = \overline{A}_{\text{int}} \cup X_{\overline{A}}$, and $\Sigma_{\text{ext}}$ is a spacelike hypersurface exterior to $\sigma_{\text{ext}}$.
Note that for $\sigma_{\text{ext}}$, we have defined the surface ``enclosed'' by $\sigma_{\text{ext}}$ to be the exterior of $\sigma_{\text{ext}}$:
\begin{equation}
    \partial \Sigma_{\text{ext}} = \sigma_{\text{ext}}.
\end{equation}

This procedure gives the interior portion of the new leaf to be
\begin{equation}
    \sigma'_{\text{int}} = \overline{A}_{\text{int}} \cup \Gamma_{\overline{A}}.
\end{equation}
We assume that for a TN representing a state with a semiclassical bulk, the RT formula with quantum corrections can be applied to the state on this new surface as well.
The process described here can be repeated multiple times, leading to a similar sequence of coarse-grained states as before.

Note that the assumption of the RT formula continuing to hold is nontrivial given that generic bulk states break the isometric property of the TN.
However, we will only need to assume that the RT formula holds for infinitesimal subregions and their complements, which gives results that are consistent with the coarse-graining interpretation suggested here.
As discussed in Sec.~\ref{sec:framework}, one could also consider bulk matter with a large central charge so that the non-isometric behavior appears only at subleading order in $1/c$ for reasonable bulk states~\cite{Bao:2018pvs,Bao:2019fpq}.


\section{Coarse-Graining and Quantum Flow}
\label{sec:main}

\subsection{Definition}

Having found the procedure in TNs, we can now look for a continuum version in semiclassical gravity.
Given the framework established in Sec.~\ref{sec:framework}, we can locate the Q-screen for a given state and start a coarse-graining procedure analogous to that discussed in Sec.~\ref{sec:motive}.

As described in Sec.~\ref{sec:review}, we consider an infinitesimal subregion on the interior portion $\sigma_{\text{int}}$ of the original leaf $\sigma$ and reduce the accessible spacetime region to the complementary quantum entanglement wedge $\text{QEW}(\overline{A})$, which is determined by the minimal QES $\Gamma_{\overline{A}}$ of $\overline{A}$ such that $\text{QEW}(\overline{A}) = D(\Sigma_{\overline{A}})$, where $\partial \Sigma_{\overline{A}} = \overline{A} \cup \Gamma_{\overline{A}}$.
Note that in a general spacetime, the global description includes an exterior portion outside $\sigma_{\text{ext}}$.
Thus, the complement of an infinitesimal subregion $A \subset \sigma_{\text{int}}$ on the leaf is $\overline{A} = \overline{A}_{\text{int}} \cup \sigma_{\text{ext}}$, and the bulk entropy term $S_{\text{bulk}}$ of the generalized entropy is given by the von~Neumann entropy of $\Sigma_{\overline{A}} = \Sigma_{\overline{A},\text{int}} \cup \Sigma_{\text{ext}}$.
The necessity of including the region exterior to $\sigma_{\text{ext}}$ can be argued from complementary recovery in pure states.

Considering many such infinitesimal subregions $A_p$ on $\sigma_{\text{int}}$ centered around points $p$ as in Eq.~(\ref{eq:intersect}), we can sequentially reduce the accessible spacetime region to
\begin{align}\label{eq:intersectquantum}
    R(\sigma) = \cap_{p} \, \text{QEW}(\overline{A_p}),
\end{align}
which leads to a renormalized leaf $\sigma_1$ such that $D(\Sigma_1) = R(\sigma)$, where $\partial \Sigma_1 = \sigma_1$.
This yields a new boundary state in a smaller effective Hilbert space defined on $\sigma_1$.
As we show in Appendix~\ref{app:QES}, the convexity of the original leaf implies that the corresponding renormalized leaf is also convex, which ensures that the coarse-graining prescription can be repeatedly applied.
The preservation of convexity also means that $S(A)$ of a subregion $A$ of a renormalized leaf obtained using Eq.~(\ref{eq:QES}) satisfies properties needed for it to be interpreted as the von~Neumann entropy of the density matrix of the region.

In the continuum limit, the behavior of QESs anchored to small subregions can be studied analytically.
While the von~Neumann entropy can in general show complicated behaviors as the subregion is varied, for an infinitesimal subregion we may expect that such behaviors arise only from physics at scales much larger than the size of the subregion.
It is then reasonable to assume that the change of the entropy of the subregion, as well as that of the complement, can be approximated by the volume integral of some density function.
With this assumption, and reasonable smoothness assumptions about the spacetime and subregions, we show in Appendix~\ref{app:flowprf} that the resulting QESs are such that the deepest point lies in a universal normal direction to the leaf given by
\begin{equation}\label{eq:flowdirn}
    s = \frac{1}{2} \left( \Theta_k l + \Theta_l k \right),
\end{equation}
as long as the relevant QESs exist.
Here, $\{ k^\mu, l^\mu \}$ are the future-directed null vectors orthogonal to $\sigma_{\text{int}}$, normalized such that $k \cdot l = -2$, and $\Theta_{k,l}$ are the corresponding quantum expansions.
Here, $\Theta_{k,l}$ are computed by varying $S_\mathrm{gen}$ as defined in Eq.~(\ref{eq:def-S_gen}).
This is sufficient to find the location of $\sigma_{1,\text{int}}$, and hence of $\sigma_1$, after fixing relative normalizations for the size of subregions considered on different portions of the leaf.
For convenience, we will choose the normalizations such that the resulting flow equation takes the simplest form.
Other possibilities will be discussed in Sec.~\ref{subsubsec:free}.

Following the procedure described above, we can derive a flow equation, which generalizes the Lorentzian mean curvature flow in Eq.~(\ref{eq:classicalflow}) to include bulk quantum corrections:
\begin{equation}\label{eq:quantumflow}
    \frac{dx^\mu}{d\lambda}=\frac{1}{2}(\Theta_k l^{\mu} + \Theta_l k^{\mu}),
\end{equation}
where $x^\mu$ are the embedding coordinates of the interior portion $\sigma_{\text{int}}(\lambda)$ of the renormalized leaves $\sigma(\lambda) = \sigma_{\text{int}}(\lambda) \cup \sigma_{\text{ext}}$, parameterized by $\lambda$, and $\Theta_{k,l}$ represent the quantum expansions of $\sigma(\lambda)$ at $x^\mu$.
The resulting sequence of $\sigma_{\text{int}}(\lambda)$ spans a codimension-1 quantum-corrected holographic slice as shown in Fig.~\ref{fig:slice}.

\begin{figure}[t]
\centering
\includegraphics[clip,width=0.7\columnwidth]{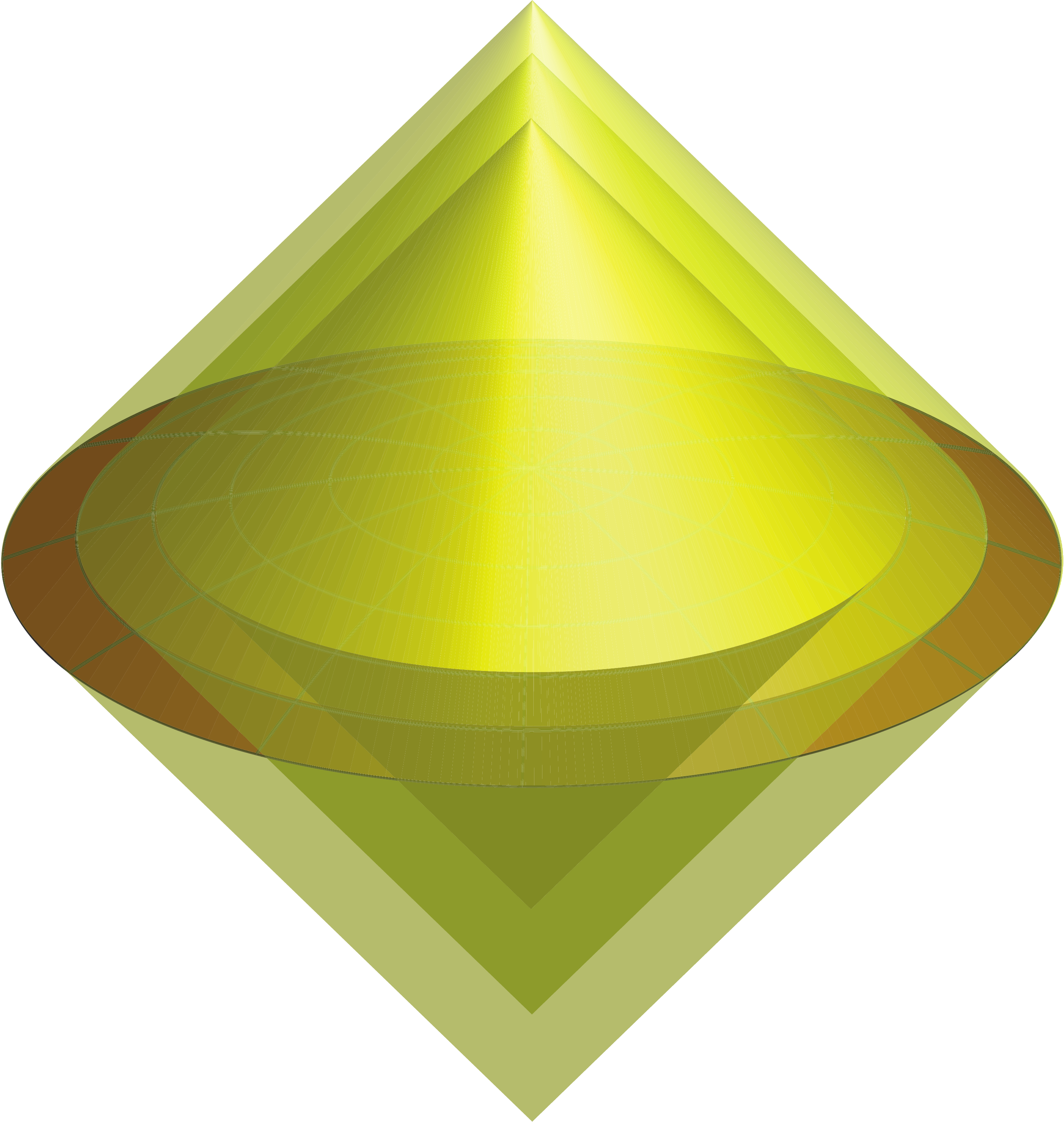}
\caption{A sequence of renormalized leaves $\sigma_{\text{int}}(\lambda)$ obtained by solving the flow equation in Eq.~(\ref{eq:quantumflow}) spans a codimension-1 quantum-corrected holographic slice.
Each leaf represents the domain of dependence of a spacelike surface $\Sigma_{\text{int}}(\lambda)$ with $\partial\Sigma_{\text{int}}(\lambda) = \sigma_{\text{int}}(\lambda)$.}
\label{fig:slice}
\end{figure}

\subsubsection{Possibility of appearance of disconnected leaf portions}
\label{subsubsec:excl}

While performing the coarse-graining as described above, it may occur that the minimal QES $\Gamma_{\overline{A_p}}$ for $\overline{A_p}$ in Eq.~(\ref{eq:intersectquantum}) becomes non-infinitesimal.
In particular, if there is a bulk region surrounded by a surface $X$ of area $\mathcal{A}(X)$ and whose entanglement with the exterior of $\sigma_{\text{ext}}$ exceeds $\mathcal{A}(X)/4 G_N$ on a given spacelike slice containing $\sigma_i = \sigma_{\text{int},i} \cup \sigma_{\text{ext}}$, then the quantum minimal surface $\chi(\overline{A_p})$ on it has a disconnected component surrounding the region.
If such a disconnected component remains after the maximization over all the spacelike slices, then the minimal QES $\Gamma(\overline{A_p})$ does have a disconnected component, and as a consequence $R(\sigma)$ will have a ``hole'' such that $\partial R(\sigma) \supset X$.
This makes the new leaf $\sigma_{\text{int},i+1}$ have a disconnected component $X$ ($\subset \sigma_{\text{int},i+1}$), in addition to the portion infinitesimally close to $\sigma_{\text{int},i}$.
The region surrounded by $X$ is thus excluded from the flow afterward.

When it first appears, an excluded region, and hence the disconnected component of the leaf associated with it, is small.
This appearance cannot be seen just by solving the flow equation, although our coarse-graining procedure itself captures the occurrence of this phenomenon.
After its appearance, the disconnected component of the leaf also flows generally, following the flow equation.
This makes the hole of the spacetime larger, which may eventually collide with the component arising from the continuous inward motion of the original leaf portion $\sigma_{\text{int}}$.

\subsection{Properties}

We now illustrate some of the salient properties of the flow, showing the consistency of it being interpreted as coarse-graining.

\subsubsection{Monotonicity of generalized entropy of renormalized leaves}

In order to interpret our procedure as coarse-graining, the number of degrees of freedom must decrease monotonically with $\lambda$.
The dimension of the effective Hilbert space $\mathcal{H}_{\text{eff}}(\sigma(\lambda))$ associated with leaf $\sigma(\lambda)$ can be defined as the amount of entropy the boundary legs carry in the TN picture, implying
\begin{align}
    \ln|\mathcal{H}_{\text{eff}}(\sigma(\lambda))| &= \frac{\mathcal{A}(\sigma(\lambda))}{4G_N} + S_{\text{bulk}}(\Sigma(\lambda))
    \nonumber\\
    &= S_{\text{gen}}(\sigma(\lambda)),
\label{eq:dim-H_eff}
\end{align}
where $|\mathcal{H}|$ represents the dimension of $\mathcal{H}$, and $\Sigma(\lambda)$ is a bulk codimension-1 spacelike surface bounded by $\sigma(\lambda) = \sigma_{\text{int}}(\lambda) \cup \sigma_{\text{ext}}$, i.e.\ $\Sigma(\lambda) = \Sigma_{\text{int}}(\lambda) \cup \Sigma_{\text{ext}}$.
We thus find that the condition for the decrease of the degrees of freedom is the same as the statement that the generalized entropy of the renormalized leaf $\sigma(\lambda)$ decreases monotonically with $\lambda$.
We now prove this in a manner similar to Ref.~\cite{Nomura:2018kji}.

We have defined the evolution vector $s$, which is tangent to the holographic slice $\Upsilon$ and radially evolves the interior leaf portion inward:
\begin{equation}
    s = \frac{1}{2} \left( \Theta_k l + \Theta_l k \right),
\end{equation}
where the associated quantum expansions satisfy
\begin{equation}
    \Theta_s = \Theta_k \Theta_l \leq 0,
\end{equation}
as shown in Appendix~\ref{app:QES}.

Consider a point $p$ on the leaf portion $\sigma_{\text{int}}(\lambda)$ and the $s$ vector orthogonal to $\sigma_{\text{int}}(\lambda)$ at $p$.
Next consider an infinitesimal patch of area $\delta A$ around $p$. 
As we flow along $s$ by a small amount, the rate at which $S_{\text{gen}}(\sigma(\lambda))$ changes is determined by the quantum expansion $\Theta_s$ as
\begin{equation}
    \delta S_{\text{gen}} \propto \Theta_s \delta \mathcal{A} \leq 0.
\end{equation}
This implies that the contribution to $S_{\text{gen}}(\sigma(\lambda))$ from the inward flow of any infinitesimal patch is negative, and hence $S_{\text{gen}}(\sigma(\lambda))$ must decrease with $\lambda$.

The argument above relies on the flow equation.
However, as shown in Section~\ref{subsubsec:excl}, it is possible that on coarse-graining, we obtain a new disconnected component of $\sigma_{\text{int}}(\lambda)$.
While the appearance of such a component cannot be described by the flow equation, it comes with a negative contribution to the generalized entropy of the renormalized leaf.
Thus, even on including this effect, we find that $S_{\text{gen}}(\sigma(\lambda))$ decreases with $\lambda$.

\subsubsection{Containment of subregion flow}

Consider the situation where we apply the holographic slice construction only to a finite subregion $A$ of the leaf portion $\sigma_{\text{int}}$.
This yields a sequence of renormalized leaves given by $\sigma(\lambda) = A(\lambda) \cup \overline{A}$.
Here, $A(\lambda)$ represents a sequence of subregions that arise from the radial evolution of A.

Now, because of entanglement wedge nesting
\begin{equation}
    \text{QEW}(\overline{A}) \subset \text{QEW}(\sigma(\lambda))
\label{eq:nesting}
\end{equation}
for arbitrary $\lambda$, since $\overline{A} \subset \sigma(\lambda)$ for all $\lambda$.
Here, $\text{QEW}(\overline{A})$ and $\text{QEW}(\sigma(\lambda))$ are determined by the corresponding minimal QESs.
This implies that the boundary of $\text{QEW}({\overline{A}})$ acts as an extremal surface barrier for the flow of $A(\lambda)$. 
In particular, $A(\lambda)$ remains outside $\text{QEW}({\overline{A}})$  for all $\lambda$.

Incidentally, if there is another QES anchored to $\partial \overline{A}$ which lies outside $\text{QEW}({\overline{A}})$, then $A(\lambda)$ would not be able to go beyond this non-minimal extremal surface.

\subsubsection{Remaining freedom}
\label{subsubsec:free}

In general, the proof in Appendix~\ref{app:flowprf} allows us to fix the direction of the flow at each point of $\sigma_{\text{int}}(\lambda)$ to be the vector $s$ as discussed in Eq.~(\ref{eq:flowdirn}).
However, there is no canonical choice of normalization, reflecting the arbitrariness of choosing relative sizes of subregions for different points on $\sigma_{\text{int}}(\lambda)$.
The ratio of these sizes must stay finite in the continuum limit, and yet it can still lead to inequivalent flow equations parameterized as
\begin{equation}
    \frac{dx^\mu}{d\lambda}=\alpha(y^i,\lambda)(\Theta_k l^{\mu} + \Theta_l k^{\mu}),
\end{equation}
where $y^i$ represents the tangential coordinates on $\sigma_{\text{int}}(\lambda)$, and $\alpha(y^i,\lambda) > 0$.
These flow equations in general result in different holographic slices, which are all gauge equivalent by the equations of motion.
By choosing subregions of the same characteristic size at all $p$, we can fix the preferred normalization that leads to Eq.~(\ref{eq:quantumflow}).

We note that this provides a natural gauge choice motivated by holography; the spacetime inside the holographic screen, which is now the Q-screen $H'$, is parameterized by $\lambda$, $y_i$, and $t$, where $t$ is a time parameter on the holographic screen giving a sequence of leaves at different times.
(If disconnected components of $\sigma_{\text{int}}(\lambda)$ appear during the flow, then we must extend $y_i$ to incorporate those components.)

An alternative choice for the normalization is to take $\lambda$ to be the proper length along the trajectory $p(\lambda)$ of a point on $\sigma_{\text{int}}(\lambda)$.
Here, the trajectory is defined such that if a point on $\sigma_{\text{int}}(\lambda + d\lambda)$ is located on the 2-dimensional surface orthogonal to $\sigma_{\text{int}}(\lambda)$ at $p(\lambda)$, then it is regarded as the ``same'' point as $p(\lambda)$, i.e.\ $p(\lambda + d\lambda)$.
This provides another natural gauge choice motivated by holography.

\subsection{End of the Flow}
\label{subsec:end}

The quantum flow procedure described above provides a way to probe a spacetime inside the holographic screen by following the holographic slice inward.
A key qualitative feature of the spacetime is given by how and where the holographic slice ends.

\begin{figure}[t]
  \subfloat[]{%
  \includegraphics[clip,width=0.8\columnwidth]{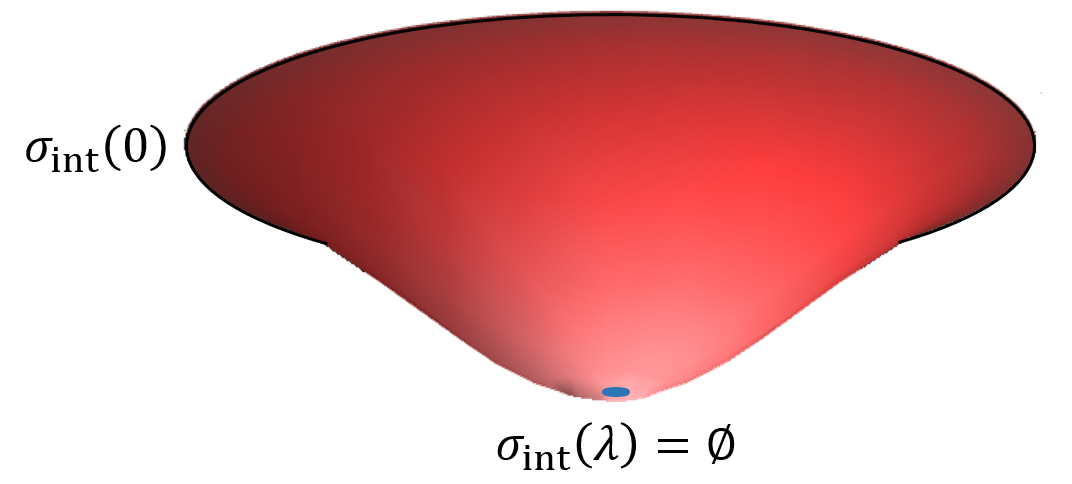}
  \label{fig:slice1}}

  \subfloat[]{%
  \includegraphics[clip,width=0.8\columnwidth]{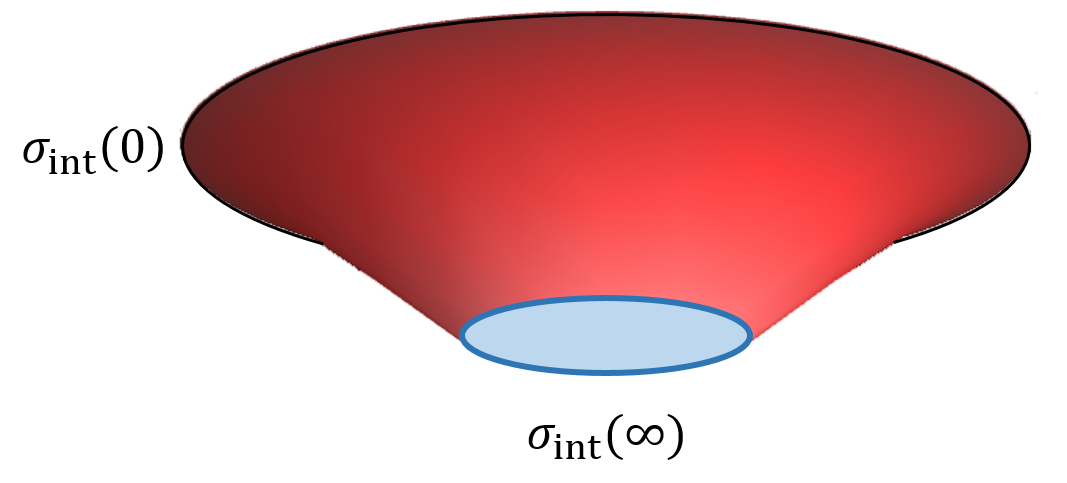}%
  \label{fig:slice2}}

  \subfloat[]{%
  \includegraphics[clip,width=0.8\columnwidth]{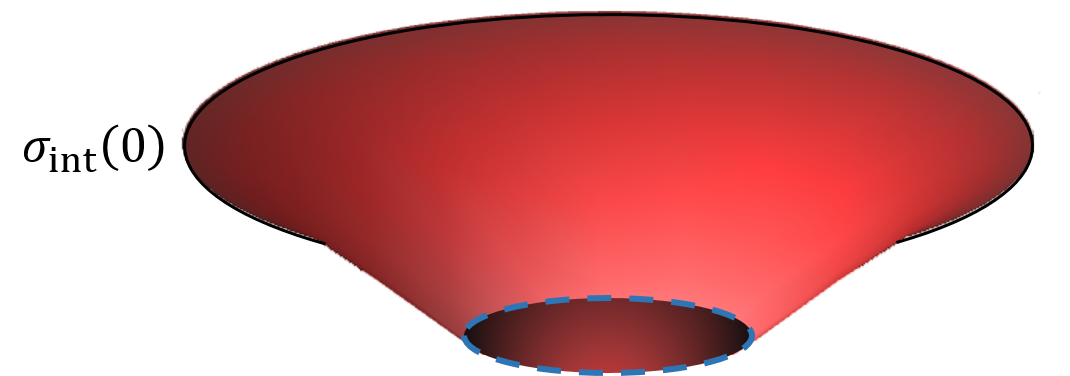}%
  \label{fig:slice3}}
\caption{Three possible ways in which the holographic slice can end.}
\label{fig:ending}
\end{figure}

The holographic slice can end in one of the three possible ways:
\begin{itemize}
    \item[(i)] The slice ends at an empty surface. This can occur simply such that $\sigma_{\text{int}}(\lambda)$ keeps moving inward, and the slice is capped off at a point, as shown in Fig.~\ref{fig:slice1}.
    Alternatively, as discussed in Sec.~\ref{subsubsec:excl}, disconnected components of $\sigma_{\text{int}}(\lambda)$ may appear during the flow, which then grow outward and coalesce with the original component of $\sigma_{\text{int}}(\lambda)$ moving inward, ending up with an empty surface.
    
    \item[(ii)] The slice asymptotes to a QES as shown in Fig.~\ref{fig:slice2}.
    As the interior portion $\sigma_{\text{int}}(\lambda)$ of renormalized leaves approaches the QES, the flow slows down because $\Theta_k, \Theta_l \rightarrow 0$.
    Note that a QES homologous to the initial leaf portion $\sigma_{\text{int}}(0)$---even if it is non-minimal---acts as a barrier which cannot be crossed as we flow in.

    \item[(iii)] The slice terminates abruptly as shown in Fig.~\ref{fig:slice3}.
    This occurs when the minimal QES associated with $\overline{A}$, the complement of an infinitesimal subregion $A$, becomes non-infinitesimal.
    At this point, we need to terminate the flow.
\end{itemize}
It is worth mentioning that while some of these cases have classical analogues, the second possibility of (i) and the case (iii) are exclusive to the quantum flow.

\subsection{Example}

Various examples for the classical flow equation of Eq.~(\ref{eq:classicalflow}) were discussed in Ref.~\cite{Nomura:2018kji}.
In many situations, the minimal QES is a small perturbation to the classical HRT surface, and accordingly the quantum corrected flow equation in Eq.~(\ref{eq:quantumflow}) results in a holographic slice that is perturbatively close to the classical holographic slice.
There are, however, cases in which the two flows are significantly different.
Here we illustrate an example of these:\ a black hole formed from collapse.

In the classical case, it was found that the holographic slice stays close to the horizon for a long time until it reaches the matter forming the black hole~\cite{Nomura:2018kji}.
It is then capped off to form a complete Cauchy slice of the spacetime as seen in Figs.~5 and 6 of Ref.~\cite{Nomura:2018kji}.
How is this modified at the quantum level?

Far from the black hole horizon, the flow is largely unaffected.
It is, however, significantly modified once we approach the horizon.
As the black hole evaporates, there are Hawking modes that escape to the region exterior to $\sigma_{\text{ext}}$, denoted $R$, leaving behind their interior partners entangled with them.
As the leaf portion $\sigma_{\text{int}}$ is moved inward by the flow, its classical area decreases but the entropy contribution from the Hawking partners increases.
About a Planck distance inside the horizon, the two effects compete with each other, resulting in a QES where the flow ends.
The mechanism by which the QES emerges here is identical to the one that appeared in a specific example in Ref.~\cite{Hartman:2020swn}.%
\footnote{Recently, Ref.~\cite{Hartman:2020khs} appeared which found such a non-trivial QES in cosmological spacetimes.
 The coarse-graining flow would end at the QES in these situations as well.}
Thus, after including bulk quantum corrections, the holographic slice becomes a partial Cauchy slice of the spacetime that excludes a large portion of the interior.%
\footnote{This does not necessarily mean that the interior of the black hole is absent. It is possible that the semiclassical interior picture emerges through approximately state-independent operators acting on modes (the hard modes~\cite{Nomura:2018kia,Nomura:2019qps}) whose characteristic frequencies are larger than the local Hawking temperature~\cite{Nomura:2019dlz}.}
The same feature can be found in the case of an AdS black hole, where one could allow the black hole to evaporate by coupling the CFT to a bath.
Our coarse-graining procedure then leads to a flow that stops at the same QES as that found in Refs.~\cite{Penington:2019npb,Almheiri:2019psf}.

\begin{figure}
    \centering
    \includegraphics[clip, width = 0.65\columnwidth]{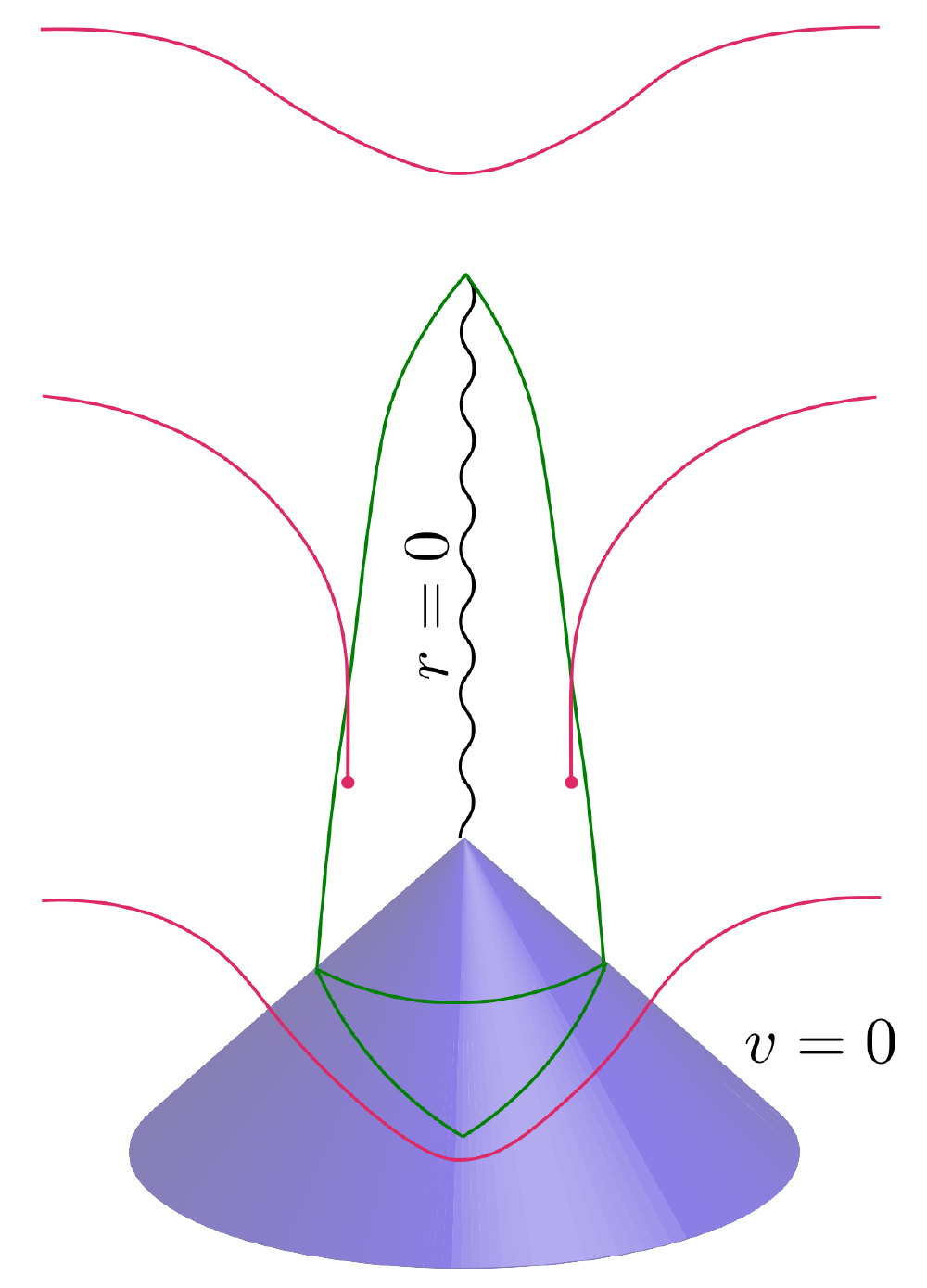}
    \caption{Eddington-Finkelstein diagram representing black hole formation and evaporation with quantum holographic slices depicted for three characteristic times.}
    \label{fig:my_label}
\end{figure}
\begin{figure}
    \centering
    \includegraphics[clip, width = 0.55\columnwidth]{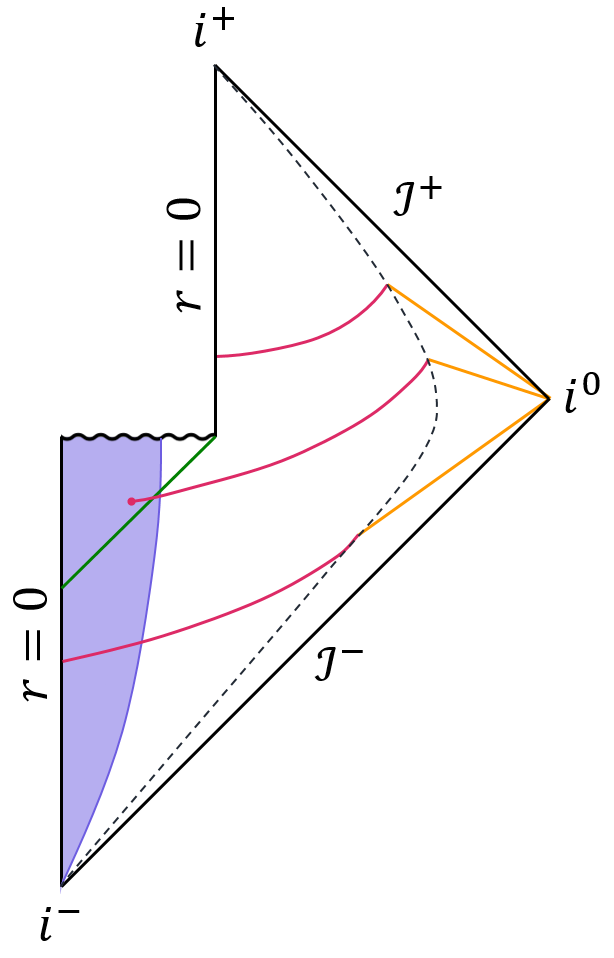}
    \caption{Penrose diagram version of Fig.~\ref{fig:my_label}.
     The region $R$ outside $\sigma_{\text{ext}}$ is depicted by orange lines.}
    \label{fig:my_label2}
\end{figure}

Another mechanism excising the black hole interior comes from the phenomenon discussed in Sec.~\ref{subsubsec:excl}.
As the black hole evaporates, there are a large number of interior partners of Hawking radiation that accumulate behind the horizon, which eventually exceed the area of the horizon at the Page time.
Hence the interior of such an old black hole would not be swept by renormalized leaves (even if the flow did not halt as described above).

The phenomenon that the holographic slice does not penetrate deep into the black hole horizon was already seen in the classical case.
There is, however, an important difference in the quantum case.
As shown in Figs.~\ref{fig:my_label} and \ref{fig:my_label2}, 
holographic slices become partial Cauchy slices during the middle of the evolution of a black hole.
(These can be contrasted with Figs.~5 and 6 of Ref.~\cite{Nomura:2018kji}.)
This implies, in particular, that with a given time parameterization on a boundary, e.g.\ on the holographic screen, the concept of black hole formation and evaporation can be rigorously defined through the behavior of the flow discussed in Sec.~\ref{subsec:end}.


\section{Relation to Quantum Error Correction}
\label{sec:QEC}

In this section, we discuss the relation between our coarse-graining procedure and the picture that the holographic dictionary works as quantum error correction~\cite{Almheiri:2014lwa,Harlow:2016vwg,Akers:2018fow,Dong:2018seb}, in which a small Hilbert space of semiclassical bulk states is mapped isometrically into a larger boundary Hilbert space.
In our framework, this picture arises after considering a collection of states over which we want to build a low energy bulk description.

In the context of quantum error correction, one chooses the set of semiclassical bulk states that can be represented as a code subspace in the boundary theory.
In a general time-dependent spacetime, however, there is no natural choice of code subspace fixed by the bulk effective theory.
This is because degrees of freedom that appear natural on one time slice need not be in bijection with those that appear natural on a different time slice.
For example, if a single heavy particle decays into a large number of radiation particles within the causal domain of $\sigma_{\text{int}}$, then we may naturally choose a code subspace associated with the degrees of freedom of the parent particle, e.g.\ its spin, or a larger subspace determined by the coarse-grained entropy of the final state radiation.

\begin{figure}[t]
  \includegraphics[clip,width=0.8\columnwidth]{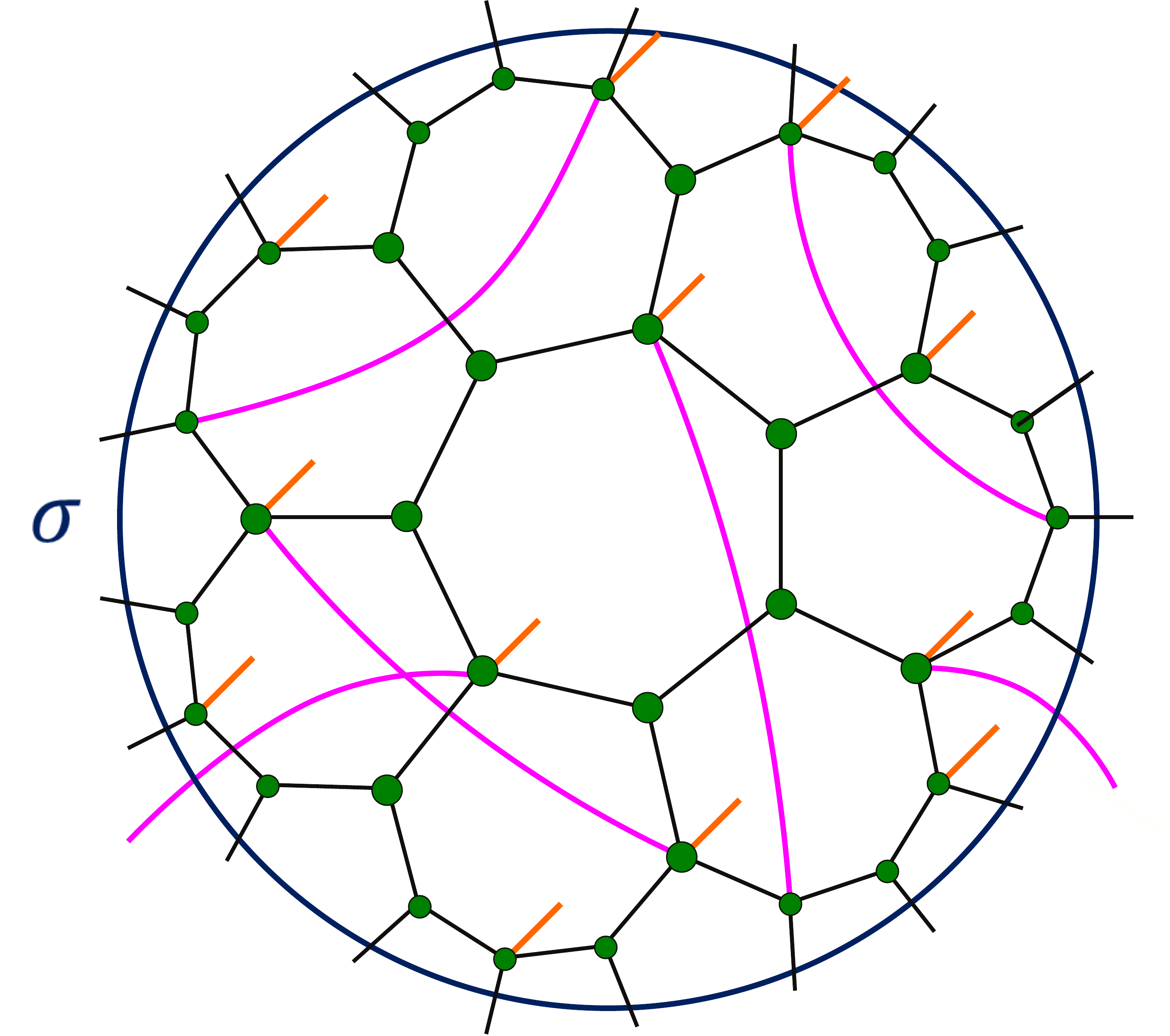}
\caption{A TN representing a collection of states has dangling legs as well as nonuniversal tensors and nonlocal bonds. }
\label{fig:dang}
\end{figure}

Our framework addresses this issue by providing a specific gauge choice given by the coarse-graining procedure.
Suppose there are a set of states giving similar geometries on their holographic slices.
Then, we can represent all these states approximately at once by a single TN, which has ``dangling'' legs so that different states in these legs correspond to different elements in the set; see Fig.~\ref{fig:dang}.
This is a choice of code subspace motivated by the coarse-graining procedure.

The introduction of dangling legs amounts to dividing bulk degrees of freedom into two classes:\ those represented by a code subspace and the rest.
Let us denote the associated Hilbert space factors by $\mathcal{H}_{\text{code}}$ and $\mathcal{H}_{\text{frozen}}$, respectively.
States in $\mathcal{H}_{\text{code}}$ correspond to the bulk degrees of freedom that we aim to reconstruct, while $\mathcal{H}_{\text{frozen}}$ is viewed as ``frozen.''
Namely, the degrees of freedom corresponding to $\mathcal{H}_{\text{frozen}}$ are treated essentially as part of the background, despite the fact that they are associated with quantum states in the conventional bulk effective field theory.

We can now define the coarse-graining in this setup, namely on a continuum analogue of a TN with dangling legs.
Specifically, we take the maximally mixed state in $\mathcal{H}_{\text{code}}$, while picking a fixed state in $\mathcal{H}_{\text{frozen}}$ determined by the network structure, i.e.\ the background geometry.
This can be thought of as considering a coarse-grained version of a generic state within the code subspace.
Indeed, the maximally mixed state plays a crucial role in AdS/CFT, in which reconstruction of an operator in the maximally mixed state is sufficient for the operator to be reconstructed on arbitrary states in the code subspace~\cite{Hayden:2018khn,Akers:2019wxj}.

\begin{figure}[t]
\centering
  \includegraphics[clip,width=0.7\columnwidth]{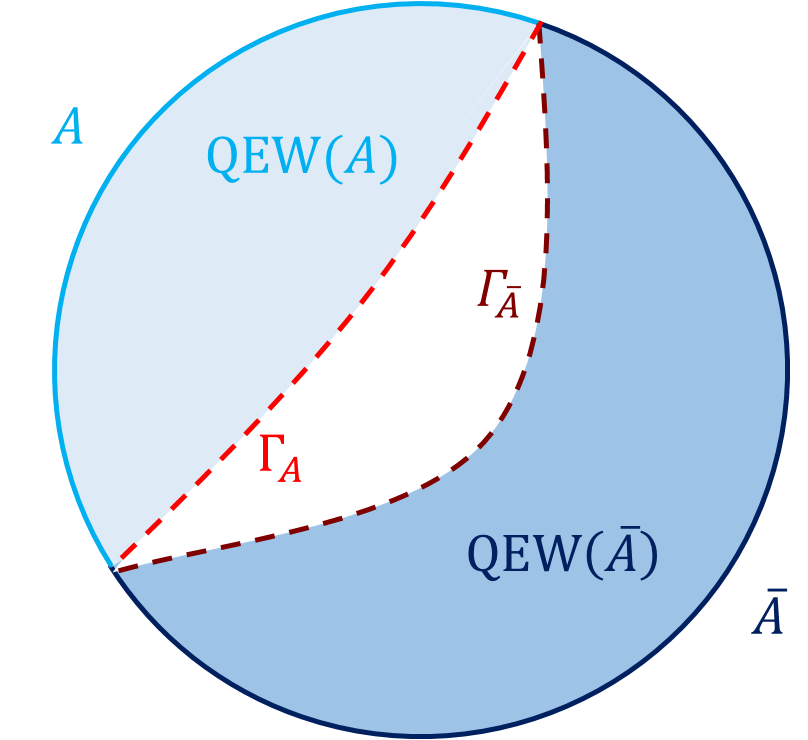}
\caption{With the state in $\mathcal{H}_{\text{code}}$ being maximally mixed, the QESs $\Gamma_A$ and $\Gamma_{\overline{A}}$ can differ.}
\label{fig:recwedg}
\end{figure}

The coarse-graining procedure then follows that in Secs.~\ref{sec:motive} and \ref{sec:main}, but this time $S_{\text{bulk}}$ receives contributions both from dangling legs, $S_{\text{code}}$, and nonlocal legs, $S_{\text{frozen}}$.
Since the QESs for $A$ and $\overline{A}$ need not agree (see Fig.~\ref{fig:recwedg}), the region between the two surfaces---often termed the no-man's land---is partially entangled with both $A$ and $\overline{A}$.
Nevertheless, using Eq.~(\ref{eq:intersectquantum}) we can obtain a picture analogous to that in Secs.~\ref{sec:motive} and \ref{sec:main}.

Once we perform the flow to obtain a renormalized leaf $\sigma(\lambda)$ appropriate to deal with the problem, e.g.\ by making $\sigma_{\text{int}}(\lambda)$ a surface surrounding the region we are interested in, then we can consider the set of all states in $\mathcal{H}_{\text{code}}$, rather than the maximally mixed state, to analyze the system in more detail.
As in the corresponding TN case, we can then interpret such coarse-grained states in two ways.

One way is to view a state in the set as defining an entangled state in the combined bulk-boundary Hilbert space
\begin{align}\label{eq:stateTN}
    \ket{\psi} \in \mathcal{H}_{\text{bulk}} \otimes \mathcal{H}_{\text{boundary}},
\end{align}
where $\mathcal{H}_{\text{bulk}}$ represents the space of bulk states in the code subspace.

Another way is to regard the set as giving an isomorphic map between the bulk Hilbert space (i.e.\ the space of dangling legs) and a subspace of the much larger boundary Hilbert space:
\begin{align}\label{eq:map}
    \{ \ket{\psi} \} &: \mathcal{H}_{\text{bulk}}\leftrightarrow \mathcal{H}_{\text{code}}\subset \mathcal{H}_{\text{boundary}}.
\end{align}
Note that in the TN picture, $\mathcal{H}_{\text{boundary}}$ consists of both local and nonlocal bonds cut by the boundary surface obtained by the flow, as well as the part associated with $\sigma_{\text{ext}}$.
This implies that the dimension of the boundary effective Hilbert space is given by
\begin{align}
    \ln|\mathcal{H}_{\text{boundary}}(\sigma(\lambda))| &= \frac{\mathcal{A}(\sigma(\lambda))}{4G_N} + S_{\text{frozen}}(\Sigma(\lambda)).
\end{align}
This can be compared with Eq.~(\ref{eq:dim-H_eff}).

In this way, holographic properties such as the HRT formula and entanglement wedge reconstruction can be naturally interpreted~\cite{Harlow:2016vwg,Cotler:2017erl}.
The interpretation is consistent with the analysis in Ref.~\cite{Hayden:2018khn} that the region reconstructable by state-independent operators---termed the reconstruction wedge in Ref.~\cite{Akers:2019wxj}---can be computed by considering $\text{QEW}(A)$ of the maximally mixed state in $\mathcal{H}_{\text{code}}$.
In this picture, our coarse-graining procedure is interpreted to produce a sequence of holographic encoding maps parameterized by $\lambda$, each of which can be viewed as a holographic duality of the form in Eq.~(\ref{eq:map}).


\section{Conclusions}
\label{sec:concl}

In this paper, we have generalized the holographic coarse-graining procedure described in Ref.~\cite{Nomura:2018kji} to include bulk quantum corrections.
Interestingly, the generalization involves promoting classical expansions $\theta$ to quantum expansions $\Theta$ as has been found in many other examples~\cite{Bousso:2015eda,Bousso:2015mna,Wall:2018ydq,Bousso:2019dxk,Bousso:2019var}.
We have demonstrated that the flow equation obtained in the bulk has all the properties consistent with an interpretation as a coarse-graining process in the holographic theory.
Our procedure also gives a way in which the region exterior to the holographic screen is treated at the quantum level.
It would be interesting to explicitly understand the detailed coarse-graining procedure from a boundary theory perspective.

\acknowledgments

We thank Nico Salzetta and Arvin Shahbazi-Moghaddam for discussions. 
This work was supported in part by the Department of Energy, Office of Science, Office of High Energy Physics under contract DE-AC02-05CH11231 and award DE-SC0019380 and in part by World Premier International Research Center Initiative (WPI Initiative), MEXT, Japan.

\appendix

\section{Convexity of Renormalized Leaves}
\label{app:QES}

\begin{definition}
\label{def:cvxset}
On a spacelike slice $\Sigma$, a compact set $S$ is defined to be convex if the quantum minimal surface $\chi_A$ anchored to the boundary $\partial A$ of a codimension-2 region $A \subset S$ is such that $\forall A$, $\chi_A \subset S$.
Here, the quantum minimal surface $\chi_A$ is defined as the surface on $\Sigma$ which minimizes the generalized entropy $S_\mathrm{gen}$ for the region on $\Sigma$ bounded by $\chi_A \cup A$.
\end{definition}

\begin{definition}
\label{def:cvxbdry}
A codimension-2 compact surface $\sigma$ is called a convex boundary if on every codimension-1 spacelike slice $\Sigma$ such that $\sigma \subset \Sigma$, the closure of the interior of $\sigma$ is a convex set.
\end{definition}

\begin{theorem}
\label{thm:cvx}
If $\sigma$ is a convex boundary, then for any subregion $A \subset \sigma$, $A \cup \Gamma_A$ is also a convex boundary.
\end{theorem}

\begin{proof}
Suppose that for some spacelike slice $\Sigma$, which contains $A \cup \Gamma_A$,  there is $B$ in the closure of the interior of $A \cup \Gamma_A$ such that the unique quantum minimal surface $\chi_B$ goes outside $A \cup \Gamma_A$. 

Since $\sigma$ is assumed to be a convex boundary, $\chi_B$ cannot go outside $\sigma$, i.e. it cannot cross $A$. 
Thus, $\chi_B$ must cross $\Gamma_A$ as shown in Fig.~\ref{fig:cvx}. 

\begin{figure}[t]
    \includegraphics[clip,width=0.6\columnwidth]{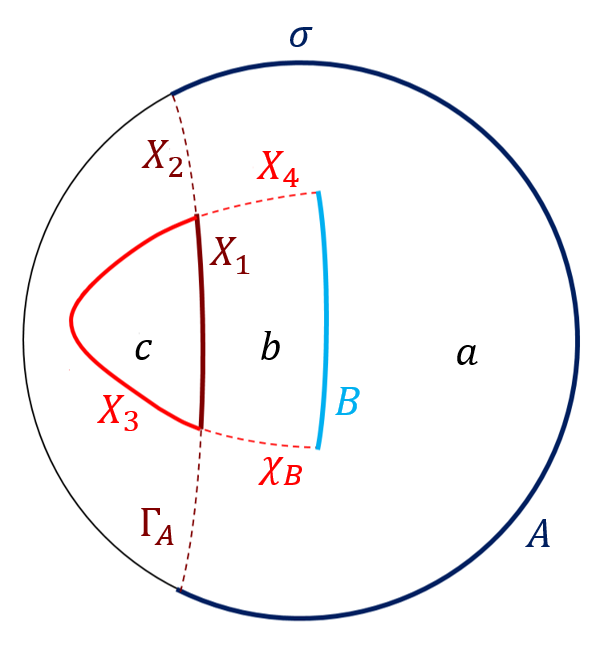}
    \caption{The quantum minimal surface $\chi_B$ crossing the quantum minimal surface $\Gamma_A$.}
    \label{fig:cvx}
\end{figure}

As $\Gamma_A = X_1 \cup X_2$ is the quantum minimal surface for $A$,
\begin{multline}
\label{eq:mincon1}
    \frac{\mathcal{A}(X_1)}{4G} +\frac{\mathcal{A}(X_2)}{4G} +
    \frac{\mathcal{A}(A)}{4G} + S_{\mathrm{bulk}}(ab) \leq \\
    \frac{\mathcal{A}(X_2)}{4G} + \frac{\mathcal{A}(X_3)}{4G} +
    \frac{\mathcal{A}(A)}{4G} + S_{\mathrm{bulk}}(abc),
\end{multline}
where the right-hand side corresponds to the surface $X_A = X_2 \cup X_3$.

Also, $\chi_B = X_3 \cup X_4$ is the unique quantum minimal surface for $B$, so 
\begin{multline}\label{eq:mincon2}
    \frac{\mathcal{A}(X_3)}{4G} +\frac{\mathcal{A}(X_4)}{4G} +
    \frac{\mathcal{A}(B)}{4G} +
    S_{\mathrm{bulk}}(bc) < \\
    \frac{\mathcal{A}(X_1)}{4G} + \frac{\mathcal{A}(X_4)}{4G} +
    \frac{\mathcal{A}(B)}{4G}+ S_{\mathrm{bulk}}(b),
\end{multline} 
where the right-hand side corresponds to the surface $X_B = X_1 \cup X_4$.

Combining Eqs. (\ref{eq:mincon1}) and (\ref{eq:mincon2}), we have
\begin{equation}
    S_{\mathrm{bulk}}(ab) + S_{\mathrm{bulk}}(bc) < S_{\mathrm{bulk}}(b) + S_{\mathrm{bulk}}(abc),
\end{equation}
which contradicts strong subadditivity. 
Thus, the quantum minimal surface $\chi_B$ crossing the quantum minimal surface $\Gamma_A$ results in a contradiction. 
\end{proof}

By repeatedly applying this theorem, we can show that if the initial leaf has the interior portion $\sigma_{\text{int}}$ that is convex, then a sequential quantum flow procedure results in a convex interior portion at each step.
In the continuum limit, this sequential procedure gives us the same renormalized leaves that were obtained using the flow equation.
Thus, we conclude that if the initial leaf portion $\sigma_{\text{int}}(0)$ is a convex boundary, then the renormalized leaf portion $\sigma_{\text{int}}(\lambda)$ is also a convex boundary for all $\lambda > 0$.

\begin{lemma}
\label{lem:cvx}
Consider a slice $\Sigma$ and a compact set $S \subset \Sigma$.
If $S$ is convex then $\Theta_\Sigma(\partial S) \leq 0$.
Here, $\Theta_\Sigma(\partial S)$ is the trace of the quantum extrinsic curvature of $\partial S$ embedded in $\Sigma$ for the normal pointing inward.
\end{lemma}

\begin{proof}
Suppose that $\Theta_\Sigma(\partial S) > 0$ somewhere on $\partial S$.
One can explicitly construct minimal surfaces that are outside $S$ by considering small enough subregions anchored to this portion of $\partial S$.
\end{proof}

Let the future-directed null vectors orthogonal to a codimension-2 spacelike surface $\sigma$ be $k$ and $l$, which we normalize as $k\cdot l = -2$.

\begin{theorem}
\label{thm:cvxbdry}
If $\sigma$ is a convex boundary, then the quantum expansions in the inward direction $\Theta_k$ and $\Theta_{-l}$ are both non-positive.
\end{theorem}
 
\begin{proof}
Consider some spacelike slice $\Sigma$. The inward normal $n$ on $\Sigma$ is given by
\begin{equation}
    n = \alpha k - \beta l
\end{equation}
with $\alpha, \beta \geq 0$.

Suppose $\Theta_k > 0$.
Now we can choose a slice $\Sigma$ such that $\Theta_\Sigma(\sigma) > 0$ by taking $\alpha \gg \beta$.
Thus, $\sigma$ is not a convex boundary because the closure of its interior is not convex on $\Sigma$ due to Lemma~\ref{lem:cvx}.
A similar argument holds if $\Theta_l < 0$.
\end{proof}

When we discuss the convexity of the leaf $\sigma(\lambda) = \sigma_{\mathrm{int}}(\lambda) \cup \sigma_{\mathrm{ext}}$, we treat $\sigma_{\mathrm{ext}}$ as a single unit on $\sigma$ which cannot be further divided into subregions.
Thus, $\sigma_{\mathrm{ext}}$ is included or excluded as a whole when we consider any boundary subregion $A$.

\section{Derivation of the Flow Equation}
\label{app:flowprf}

Consider a codimension-2, closed, achronal surface $\sigma$ in an arbitrary $(d + 1)$-dimensional spacetime $M$.
Suppose $\sigma$ is a convex boundary.
We assume that both $M$ and $\sigma$ are sufficiently smooth so that variations in the spacetime metric $g_{\mu\nu}$ and induced metric on $\sigma$, denoted by $h_{ij}$, occur on characteristic length scales $L_g$ and $L_{\sigma}$, respectively.
We also assume that the changes of the variational entropy current density $J^{\mu}(x)$, discussed below, occur on a characteristic length scale $L_S$.

\begin{theorem}
\label{thm:flweqn}
Consider subregion $A$ of characteristic length $\delta \ll L_g, L_\sigma, L_S$ on the surface $\sigma$.
This subregion $A$ is chosen to be a $(d-1)$-dimensional ellipsoid on $\sigma$ at order $O(\delta)$. 
Then, at the leading order, the QES anchored to $\partial A$ lives on the hypersurface generated by the evolution vector%
\footnote{In this appendix, we ignore the possibility that there is no QES infinitesimally close to $A$ as $\delta \rightarrow 0$, i.e.\ the possibility (iii) discussed in Sec.~\ref{subsec:end}.}
\begin{equation}
   s = \frac{1}{2} (\Theta_k l + \Theta_l k) =  \Theta_t t -\Theta_z z
\end{equation}
normal to $\sigma$.
Here, $k$ and $l$ are future-directed null vectors orthogonal to $\sigma$ normalized as $k\cdot l = -2$, and $t$ and $z$ are vectors related to these by
\begin{equation}
\label{eq:defkl}
   k = (t+z), \hspace{15pt} l = (t-z). 
\end{equation}
\end{theorem}

\begin{proof}
Since the subregion $A$ is assumed to be an ellipsoid, we label its center point as $p$.
We can then set up Riemann normal coordinates in the local neighborhood of $p$:
\begin{equation}
    g_{\mu\nu}(x) = \eta_{\mu\nu} - \frac{1}{3} R_{\mu\nu\rho\sigma} x^{\rho} x^{\sigma} + O(x^3).
\end{equation}
In these coordinates, we are considering a patch of size $O(\delta)$ around the origin $p$ with $R_{\mu\nu\rho\sigma} \sim O(1/L_g^2)$, so at any point in this patch
\begin{equation}
    g_{\mu\nu}(x) = \eta_{\mu\nu} + O\left(\frac{\delta^2}{L_g^2}\right).
\end{equation}

Since there is still a remaining $SO(d, 1)$ symmetry that preserves the Riemann normal coordinate form of the metric, we can use these local Lorentz boosts and rotations to set $t$ and $z$ as the coordinates in the normal direction to $\sigma$ at $p$ while $y^i$ parameterize the tangential directions.
At order $O(\delta)$, the subregion $A$ is then an ellipsoid in the $y^i$ coordinates centered at the origin $p$.

In a small region around $p$, we can define a variational entropy current density that measures how $S_\mathrm{bulk}$ changes. 
More formally, let $X_A$ be a surface anchored to the boundary of $A$, or equivalently of $\overline{A}$:\ $X_A = X_{\overline{A}}$.
Let $\Xi_{\overline{A}}$ be the homology surface with $\partial\Xi_{\overline{A}} = {\overline{A}} \cup X_{\overline{A}}$; then
\begin{equation}
\label{eq:svar}
    S_\mathrm{bulk}(\Xi_{\overline{A}}) = S_0 - \int_{\mathcal{S}} J_{\mu}(x) da^{\mu},
\end{equation}
where $S_0$ is the $S_\mathrm{bulk}$ associated with the full $\sigma$, so it is independent of the choice of subregion $A$ or the surface $X_A$.
$J_{\mu}(x)$ is the aforementioned variational entropy current density which upon integrating over $\mathcal{S}$, a homology surface with boundary $\partial\mathcal{S} = A \cup X_{\overline{A}}$, determines how $S_\mathrm{bulk}(\Xi_{\overline{A}})$ differs from $S_0$. 

We now Taylor expand the entropy current density about $p$, so for any point within $O(\delta)$ distance of $p$
\begin{equation}
\label{eq:jexp}
    J_\mu(x) = \mathcal{J}_{\mu} \left(1 + O\left(\frac{\delta}{L_S}\right) \right),
\end{equation}
where $\mathcal{J}_{\mu}=J_{\mu}(0)$.
Recall that $L_S$ is the length scale of the variations of corresponding entropy variations. 

Let $K^t_{ij}$ , $K^z_{ij}$ denote the extrinsic curvature tensors of $\sigma$ for the $t$ and $z$ normals, respectively.
It follows that $K^t_{ij}, K^z_{ij} \sim O(1/L_\sigma)$.
Since $t$ and $z$ are normal to $\sigma$, the equations for the surface $\sigma$, described by $t_L(y^i)$ and $z_L(y^i)$, can be Taylor expanded in the region $A$ as
\begin{align}
    \label{eq:tleaf}
    t_L(y^i) & = -\frac{1}{2} K^t_{ij} y^i y^j + O\left(\frac{\delta^3}{L_\sigma^2}\right),\\
    \label{eq:zleaf}
    z_L(y^i) &= \frac{1}{2} K^z_{ij} y^i y^j + O\left(\frac{\delta^3}{L_\sigma^2}\right),
\end{align}
where the negative sign in the first line is due to the time-like signature of the $t$ normal.

From Eqs.~(\ref{eq:tleaf}) and (\ref{eq:zleaf}), it follows that at the leading order in $\delta$,
\begin{align}
    \label{eq:tleaf2}
    \nabla^2 t_L & =  - \eta^{ij} K^t_{ij}, \\
    \label{eq:zleaf2}
    \nabla^2 z_L &= \eta^{ij} K^z_{ij} ,
\end{align}
where $\nabla^2 = \partial^i \partial_i$. 
Note that $h_{ij}=\eta_{ij}$ at this order. 

It follows that the ratio of quantum null expansions on the surface $\sigma$ is
\begin{equation}
    \frac{\Theta_k}{\Theta_l} = \frac{\eta^{ij} (K^t_{ij} + K^z_{ij}) + 4G_N(\mathcal{J}_t - \mathcal{J}_z)}{\eta^{ij} (K^t_{ij} - K^z_{ij}) + 4G_N(-\mathcal{J}_t-\mathcal{J}_z)},
\end{equation}
or equivalently
\begin{equation}
    \frac{\Theta_t}{\Theta_z} = \frac{\Theta_k+\Theta_l}{\Theta_k-\Theta_l} = \frac{\eta^{ij} K^t_{ij} - 4G_N \mathcal{J}_z}{\eta^{ij} K^z_{ij} + 4G_N \mathcal{J}_t}.
\end{equation}
Here, we have used that the bulk entropy is given by Eq.~(\ref{eq:svar}) along with the Taylor expansion in Eq.~(\ref{eq:jexp}).

The QES $\Gamma_{\overline{A}}$ can be parameterized in a similar way using $t_E(y^i)$ and $z_E(y^i)$.
The boundary conditions satisfied by the QES are
\begin{align}
    t_E(\partial A) &= t_L(\partial A),\\
    z_E(\partial  A) &= z_L(\partial A).
\end{align}

Since the region $A$ is chosen to be an ellipsoidal region in the $y^i$ coordinates at $O(\delta)$, we have symmetry under $y^i \rightarrow - y^i$ at this order.
Consequently, the $t$ and $z$ directions are normal to the QES at the center point $(t_E(0), z_E(0), 0)$.

Let $\widetilde{K}^t_{ij}$, $\widetilde{K}^z_{ij}$ denote the extrinsic curvature tensors of the QES for the $t$ and $z$ normals, respectively.
We assume that the QES is approximately flat at lengthscale $\delta$, i.e.\ $ \widetilde{K}^t_{ij}, \widetilde{K}^z_{ij} \ll 1/ \delta $.
We will show that this assumption is self-consistent as long as the entropy current density is not too large.

We can Taylor expand $t_E(y^i)$ and $z_E(y^i)$ as
\begin{align}
    \label{eq:tqes}
    t_E(y^i) &= t_E(0) - \frac{1}{2} \widetilde{K}^t_{ij} y^i y^j,\\
    \label{eq:zqes}
    z_E(y^i) &= z_E(0) + \frac{1}{2} \widetilde{K}^z_{ij} y^i y^j.
\end{align}
Since the QES has vanishing quantum null expansion, we have at the leading order
\begin{align}
    \eta^{ij} (\widetilde{K}^t_{ij} + \widetilde{K}^z_{ij} ) + 4G_N(\mathcal{J}_t - \mathcal{J}_z) & = 0,\\
    \eta^{ij} (\widetilde{K}^t_{ij} - \widetilde{K}^z_{ij} ) + 4G_N(-\mathcal{J}_t - \mathcal{J}_z) & = 0.
\end{align}
Here, we have used the expansion in Eq.~(\ref{eq:jexp}) because any point on the QES is at most $O(\delta)$ distant from the origin $p$. 

These equations, along with Eqs.~(\ref{eq:tqes}) and (\ref{eq:zqes}), result in the following differential equations for $t_E(y^i)$ and  $z_E(y^i)$ at the leading order
\begin{align}
    \label{eq:tqes2}
    \nabla^2 t_E & =  - \eta^{ij} \widetilde{K}^t_{ij}  = - 4 G_N \mathcal{J}_z, \\
    \label{eq:zqes2}
    \nabla^2 z_E &= \eta^{ij} \widetilde{K}^z_{ij} = - 4 G_N \mathcal{J}_t.
\end{align}
Earlier, we assumed that $\widetilde{K}^t_{ij}, \widetilde{K}^z_{ij} \ll 1/\delta$, which is justified as long as $\mathcal{J}_t, \mathcal{J}_z \ll 1/(4 G_N \delta)$, which is the case if $\mathcal{J}_t, \mathcal{J}_z$ do not diverge as $\delta \rightarrow 0$.

Let us consider the quantities $\delta t = t_E - t_L$ and $\delta z = z_E - z_L$.
These satisfy the following differential equations at the leading order
\begin{align}
    \label{eq:tdel}
    \nabla^2 \delta t & =   \eta^{ij} K^t_{ij}  - 4 G_N \mathcal{J}_z, \\
    \label{eq:zdel}
    \nabla^2 \delta z &= - \eta^{ij} K^z_{ij}  - 4 G_N \mathcal{J}_t. 
\end{align}
The boundary conditions are given by
\begin{equation}
    \delta t (\partial A) = 
    \delta z (\partial A) = 0.
\end{equation}

It is now clear that at the leading order, $\delta t/\Theta_t$ and $-\delta z/\Theta_z$ satisfy the same differential equation with the same boundary conditions, since $\Theta_t$ and $\Theta_z$ can be regarded as constant at this order.
Thus,
\begin{equation}
    \frac{\delta t}{\Theta_t} = - \frac{\delta z}{\Theta_z} 
    \left( 1 + O(\delta) \right)
\label{eq:appB-final}
\end{equation}
for all points on the extremal surface.
Rewritten, the extremal surface lives on the hypersurface generated by the evolution vector $s = \Theta_t t -\Theta_z z$ normal to $\sigma$. 
\end{proof}

In Theorem~\ref{thm:flweqn} above, we have assumed that the subregion $A$ is a $(d-1)$-dimensional ellipsoid on the surface $\sigma$.
Nonetheless, the proof goes through if the subregion $A$ has a reflection symmetry $(y^1, y^2, ..., y^{d-1}) \rightarrow (-y^1, -y^2, ..., -y^{d-1}) $ about the center point $p$ at order $O(\delta)$.

In fact, we expect this theorem to hold for a more general subregion $A$ because the above proof works if we can find any point $p \in A$ such that the normal vectors to $\sigma$ at $p$ match with the normal vectors to the QES at the point corresponding to $p$ at order $O(\delta)$.
Under the condition that $\delta \ll L_g, L_\sigma, L_S$, such a point lies at the ``center,'' in the sense that the above leading-order treatment works; for example, the QES of the form of Eqs.~(\ref{eq:tqes}) and (\ref{eq:zqes}) is correctly ``anchored'' to $\partial A$ at the leading order in $\delta$.

Finally, our discussion in this appendix applies in the context of the main text to the interior portion of the leaf, $\sigma_{\text{int}}$.
The existence of the exterior portion $\sigma_{\text{ext}}$ does not change the fact that the QES of $\overline{A}$ lies on the hypersurface given by Eq.~(\ref{eq:appB-final}).

\bibliography{mybibliography}

\end{document}